\documentclass[11pt,letterpaper]{article}

\usepackage[T1]{fontenc}
\usepackage[utf8]{inputenc}
\usepackage[margin=1in]{geometry}
\usepackage{times}
\usepackage{microtype}

\usepackage{amsmath,amssymb,mathtools,amsthm}
\usepackage{bm}

\usepackage{graphicx}
\usepackage{caption}
\usepackage{subcaption}
\usepackage{booktabs}
\usepackage{multirow}
\usepackage{threeparttable}
\usepackage{arydshln}
\usepackage{makecell}
\usepackage{wrapfig}
\usepackage{float}
\usepackage{placeins}
\usepackage{needspace}
\usepackage{adjustbox}
\usepackage{enumitem}
\usepackage{xcolor}
\usepackage{xspace}
\usepackage{titletoc}

\usepackage[authoryear,round]{natbib}
\usepackage[hidelinks]{hyperref}
\usepackage[capitalize,noabbrev]{cleveref}
\hypersetup{
  pdftitle={Symmetric Composition of Anisotropic Operators for Global Subseasonal-to-Seasonal Climate Forecasting},
  pdfauthor={Ziyu Zhou, Yuchen Fang, Weilin Ruan, Tian Zhou, Shiyu Wang, James Kwok, Yuxuan Liang},
  pdfsubject={Preprint}
}

\theoremstyle{plain}

\theoremstyle{definition}

\theoremstyle{remark}

\newcommand{\icon}{\raisebox{-5pt}{\includegraphics[width=1.4em]{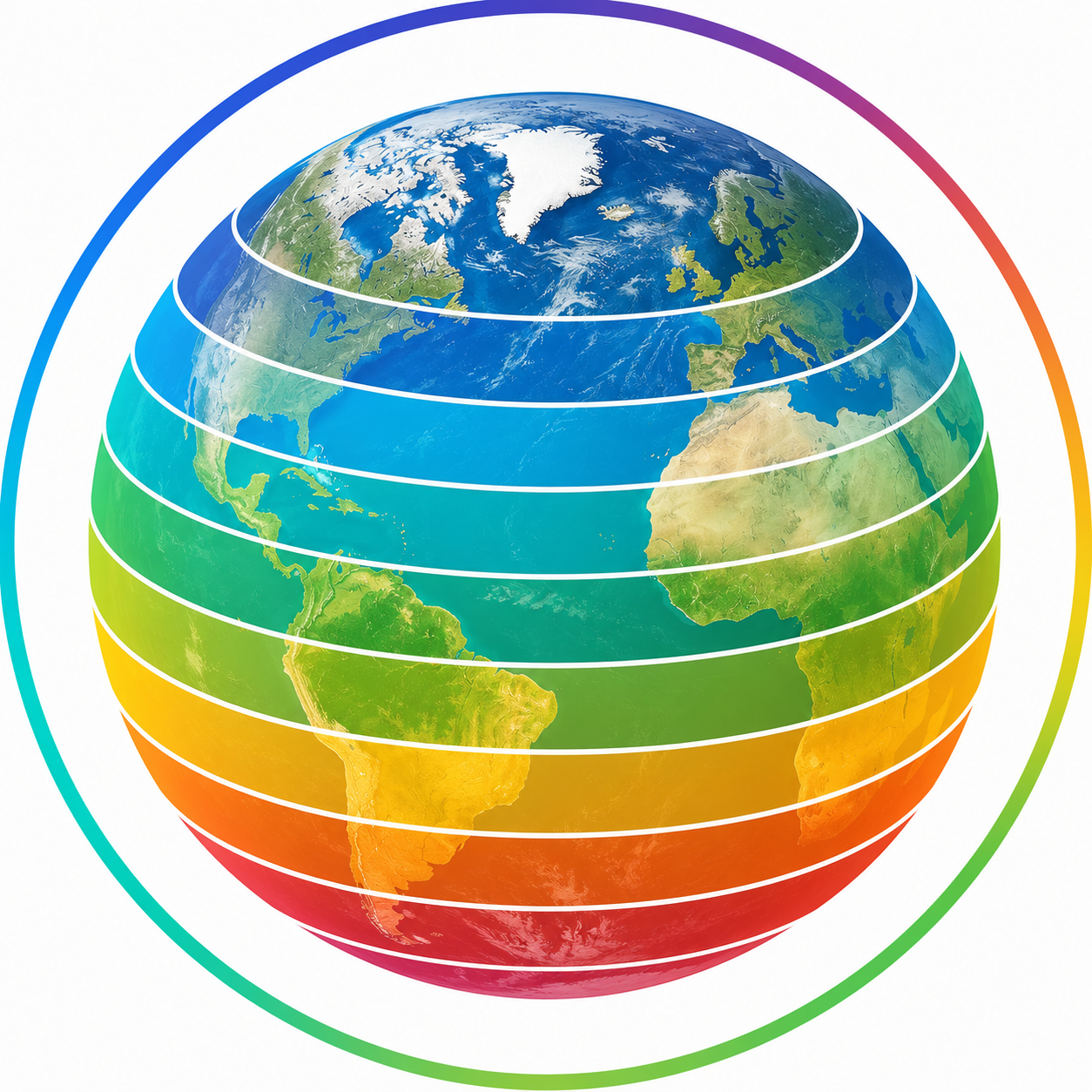}}\xspace}
\title{\icon
Symmetric Composition of Anisotropic Operators for Global Subseasonal-to-Seasonal Climate Forecasting
}

\author{
  \parbox{0.96\textwidth}{\centering
    \normalsize
    Ziyu Zhou\textsuperscript{1},
    Yuchen Fang\textsuperscript{2},
    Weilin Ruan\textsuperscript{3},
    Tian Zhou\textsuperscript{4},\\[0.3em]
    Shiyu Wang,
    James Kwok\textsuperscript{1,*},
    Yuxuan Liang\textsuperscript{6}\\[0.8em]
    \small
    \textsuperscript{1}The Hong Kong University of Science and Technology\\
    \textsuperscript{2}University of Electronic Science and Technology of China\\
    \textsuperscript{3}The Chinese University of Hong Kong \quad
    \textsuperscript{4}Ant Group \\
    \textsuperscript{6}The Hong Kong University of Science and Technology (Guangzhou)\\[0.5em]
    \textsuperscript{*}Corresponding author.\\[0.2em]
  }
}
\date{}

\begin{document}

\maketitle

\begin{abstract}
Accurate global Subseasonal-to-Seasonal (S2S) climate forecasting is critical for disaster preparedness and resource management, yet it remains challenging due to chaotic atmospheric dynamics. Despite advances in geometry-aware representations, existing methods do not specify how the zonal and meridional interactions that govern real atmospheric dynamics should be explicitly coordinated and composed, leaving an important architectural gap for S2S forecasting. In this paper, we propose \textbf{AnisoCast}, which explicitly
models anisotropic global atmospheric dynamics for accurate
S2S climate forecasting. It couples: (1) an Anisotropic Embedding strategy that tokenizes the global grid into latitudinal rings, preserving the integrity of zonal periodic structures; and (2) stacked Aniso Blocks that arrange latent Zonal and Meridional Operators in a weight-shared palindromic composition inspired by symmetric operator splitting. Extensive experiments on the ERA5 reanalysis dataset demonstrate that AnisoCast establishes a new state-of-the-art, significantly outperforming existing methods in both forecasting accuracy and computational efficiency. Code is available at \href{https://github.com/zhouziyu02/AnisoCast}{https://github.com/zhouziyu02/AnisoCast}.
\end{abstract}

\section{Introduction}
\label{sec:intro}
Global subseasonal-to-seasonal (S2S) climate forecasting aims to
predict atmospheric conditions across the Earth, from several
weeks (subseasonal) to months (seasonal) ahead.
This is vital for decision-making in various sectors such as energy scheduling, agriculture,
and disaster management~\citep{vitart2017subseasonal,white2017potential,Cirt,pnpcorrector}.
However, while data-driven models have advanced medium-range forecasting up to approximately two weeks~\citep{bi2023accurate,lam2023learning,gencast}, the performance often degrades as the lead time extends into the subseasonal range~\citep{chaosbench,telepit}.
Forecasting at the S2S scale is
even more challenging as much of the predictive information carried by the
initial atmospheric state has already faded, whereas slower
Earth system processes (such as ocean variability) generally
require longer timescales to exert a strong influence on
atmospheric evolution~\citep{mariotti2018progress,vitart2018sub,Cirt,neuralom}.

To improve S2S forecasting performance,
recent
attempts introduce inductive biases
based on spherical geometry to better capture global atmospheric dynamics~\citep{FuXi-S2S,zheng2025mesh}. For example, CirT adapts to the spherical grid structure via geometry-inspired circular patching \citep{Cirt}. EIMP uses equivariant message passing on spherical meshes~\citep{EIMP}. Spectral approaches have also been explored to integrate teleconnection modeling with spherical harmonic representations~\citep{telepit} and incorporate Fourier-based embeddings for enhancing wave pattern representations~\citep{TianQuans2s}.

Geometry-aware representations alone do not explicitly coordinate
zonal and meridional interactions, leaving a key architectural gap
in modeling anisotropic atmospheric dynamics.
Global atmospheric fields exhibit anisotropic dynamics
(Figure~\ref{fig:concept}), in that
their dynamical behavior (such as Rossby-wave propagation and the transport of heat
and moisture)
differs between the zonal
(i.e., east--west) and meridional (i.e., north--south) directions~\citep{rossby1939relation,hoskins1981steady}.
The two directions also differ geometrically. While longitude wraps around the globe
and is periodic
(e.g.,
$0^\circ$ and $360^\circ$ refer to the same meridian, so do $-180^\circ$ and $180^\circ$),
latitude is finite~\citep{oort1992physics,vallis2017atmospheric}.
However, geometric design techniques such as 2D patching~\citep{TianQuans2s},
equivariant and invariant modeling~\citep{EIMP}, and spherical harmonic aggregation~\citep{telepit}
do not specify how zonal and meridional interactions should be separately processed and composed. Encoding geometric structure therefore leaves an architectural
gap between representing directional information and explicitly organizing its processing through complementary directional operations. This motivates an explicit design of complementary zonal and meridional operations and their composition within the forecasting backbone.

\begin{wrapfigure}{r}{0.48\textwidth}
    \centering
    \vspace{-4mm}
    \includegraphics[width=\linewidth, trim=13 21 8 12, clip]{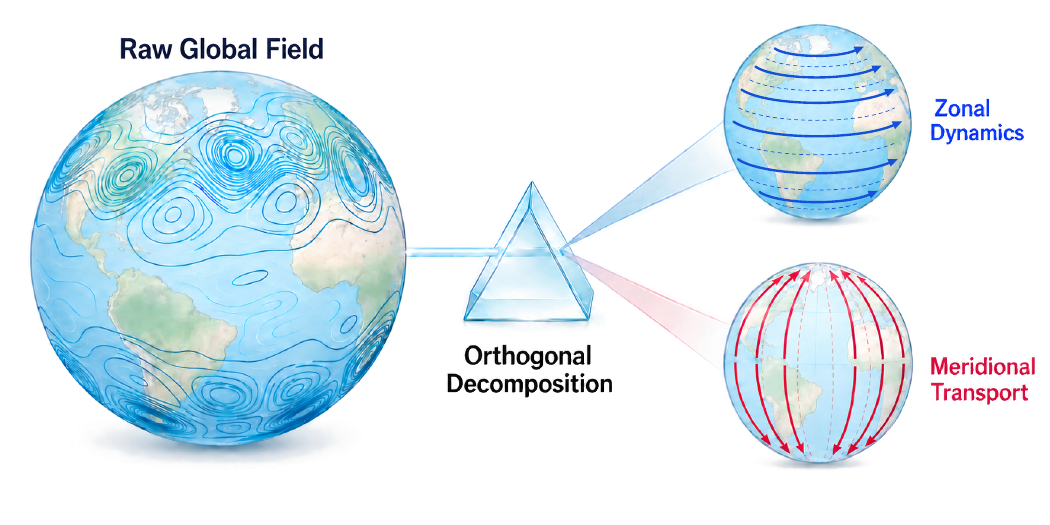}
    \vspace{-6mm}
    \caption{The anisotropic nature of global atmospheric dynamics. Global atmospheric circulation is governed by the coupling of two orthogonal physical processes:
    (i) Zonal dynamics, characterized by east--west wave propagation, and (ii) meridional transport, driven by the north-south exchange of heat and momentum.}
    \vspace{-3mm}
   \label{fig:concept}
\end{wrapfigure}

To alleviate these problems, we propose \textbf{AnisoCast},
which is designed to capture anisotropic atmospheric dynamics
through a symmetric composition of zonal and meridional
operators for effective global S2S climate forecasting.
AnisoCast first employs an \textbf{Anisotropic Embedding} strategy that tokenizes the grid into latitudinal rings, preserving the integrity of zonal periodic structures while structurally distinguishing them from meridional representations.
The \textbf{Aniso Blocks} then
use \textbf{Zonal Operators} to capture
the global spectral wave dynamics and \textbf{Meridional Operators} for the
local spatial transport,
allowing explicit modeling of their alternating interaction.
These operators induce a geometrically orthogonal
representation, thereby aligning the model with the anisotropic nature of global atmospheric dynamics. Each Aniso Block also introduces structural symmetry to mitigate the accumulation of discretization errors over the long S2S forecasting horizon.
Finally, an MLP-based \textbf{Decoder} predicts all
target variables, and the model is trained end-to-end.
Our main contributions are as follows:

\begin{itemize}[leftmargin=*]
    \item We identify that prevalent S2S forecasting models use
    isotropic modeling and highlight the importance of explicitly distinguishing
    zonal and meridional interactions in S2S forecasting architectures.
    \item We propose AnisoCast, which integrates Anisotropic Embedding with a symmetric composition of Zonal and Meridional Operators, that is theoretically and structurally aligned with long-term atmospheric dynamics at the S2S timescale.
    \item Extensive experiments on the Earth Reanalysis 5 (ERA5) dataset demonstrate that AnisoCast establishes a new state-of-the-art
    in S2S climate forecasting, outperforming existing methods in both
    accuracy and computational efficiency.
\end{itemize}

\section{Preliminaries}
\label{sec:preliminaries}

 \vspace{-2mm}
\subsection{Problem Definition}
\label{subsec:problem_definition}

\vspace{-1mm}
We consider global bi-weekly forecasting of $K$ meteorological variables defined on a latitude--longitude grid. Let $\Omega = [-90^{\circ}, 90^{\circ}] \times [-180^{\circ}, 180^{\circ}]$ be the spherical coordinate domain, discretized into $H$ latitude bands and $W$ longitude points.
We represent the spatial coordinates by $\bm{\mathcal{P}} \in \mathbb{R}^{H \times W \times 2}$ with
$\bm{\mathcal{P}}_{h,w,:} = (\varphi_h,\lambda_w) \in \Omega$,
where $\varphi_h$ and $\lambda_w$ are the latitude and longitude, respectively.
At day $t$, the global atmospheric state is a tensor $\bm{X}_t \in \mathbb{R}^{K \times H \times W }$.
In S2S benchmarking and S2S direct forecasting \citep{mouatadid2023adaptive,chaosbench,telepit,Cirt,EIMP,FuXi-S2S},
we want to learn a predictor $\mathcal{F}_{\Theta}$ that,
given an initial condition $(\bm{\mathcal{P}}, \bm{X}_{t_0})$,
directly outputs the two verification targets for weeks 3-4 and weeks 5-6:
\begin{equation}
\label{eq:problem_def}
(\hat{\bm{Y}}_{15:28},\,\hat{\bm{Y}}_{29:42})
=
\mathcal{F}_{\Theta}(\bm{\mathcal{P}},\bm{X}_{t_0}),
\end{equation}
where $\hat{\bm{Y}}_{15:28} \in \mathbb{R}^{K \times H \times W}$ and $\hat{\bm{Y}}_{29:42} \in \mathbb{R}^{K \times H \times W }$ denote the predicted averages over days $15$-$28$ (weeks 3-4) and days $29$-$42$ (weeks 5-6).
A review on weather forecasting and S2S climate forecasting can be found in Appendix~\ref{sec:related_works}.

\subsection{Operator Splitting for Non-commutative Dynamics}
\label{subsec:operator_splitting}

Atmospheric evolution can be abstracted as the dynamical system
$\frac{\partial \mathbf{u}}{\partial t}\equiv\mathcal{F}(\mathbf{u})=\mathcal{A}(\mathbf{u})+\mathcal{B}(\mathbf{u})$,
where $\mathbf{u}(t)$ is the global state vector,
$\mathcal{A}$  is the
zonal wave dynamics
component,
and $\mathcal{B}$ is
the meridional transport component~\citep{rossby1939relation,hoskins1981steady,oort1992physics}.
The
coupled flow,
governed jointly by
$\mathcal{A}$ and $\mathcal{B}$,
is written
as
$\mathcal{F}(\mathbf{u})=(\mathcal{A}+\mathcal{B})\mathbf{u}$.

Consider its spatially discretized linear analogue
$\frac{\partial \mathbf{u}}{\partial t}
=(\mathbf{A}+\mathbf{B})\mathbf{u}$,
where $\mathbf{A}$ (resp. $\mathbf{B}$) is the matrix representation of
the discretized operator for $\mathcal{A}$ (resp. $\mathcal{B}$)~\citep{mclachlan2002splitting,hairer2006geometric}.
The coupled flow
over an interval of length $\tau$
can then be written as
$\mathbf{u}(t+\tau)=\exp(\tau(\mathbf{A}+\mathbf{B}))\,\mathbf{u}(t)$,
where $\exp(\cdot)$ denotes the matrix exponential. In general, $\mathbf{A}$ and $\mathbf{B}$
do not commute
(i.e.,
$\mathbf{A}\mathbf{B}-\mathbf{B}\mathbf{A}\neq\mathbf{0}$)~\citep{lanser1999analysis,dimov2004operator}, and so
the flow cannot be factorized into separate evolutions.

Operator splitting~\citep{trotter1959product,strang1968construction,mclachlan2002splitting} approximates this coupled flow by alternating two simpler subflows. Each subflow exploits its own structure while retaining the interaction through composition.
Different splitting orders correspond to different discretizations of the underlying dynamics.
The two representative methods are:
(i) \textbf{Lie-Trotter splitting}~\citep{trotter1959product},
which is a first-order approximation of the coupled flow and applies the two subflows sequentially as:
$\mathcal{S}_{\text{Lie}}(\tau)=\exp(\tau\mathbf{B})\exp(\tau\mathbf{A})$.
Its one-step local error is $\mathcal{O}(\tau^2)$~\citep{blanes2025concise}.
(ii)
\textbf{Strang splitting}~\citep{strang1968construction}, which is
a symmetric second-order decomposition that enforces time-reversal symmetry via:
\begin{equation}
\label{eq:strang}
\mathcal{S}_{\mathrm{Strang}}(\tau)
=
\exp\!\left(\tfrac{\tau}{2}\mathbf{A}\right)
\exp(\tau\mathbf{B})
\exp\!\left(\tfrac{\tau}{2}\mathbf{A}\right).
\end{equation}
This yields an $\mathcal{O}(\tau^3)$ one-step local error for a sufficiently small $\tau$~\citep{mclachlan2002splitting}.

In long-horizon forecasting,
local approximation errors can accumulate under repeated composition.
When modeling non-commutative systems,
symmetric decomposition is often preferred
over Lie-Trotter splitting
as it typically reduces systematic bias and slows error accumulation  \citep{hairer2006geometric,blanes2025concise}.

\begin{figure*}[t]
\centering
\includegraphics[width=\linewidth, trim=20 11 15 15, clip]{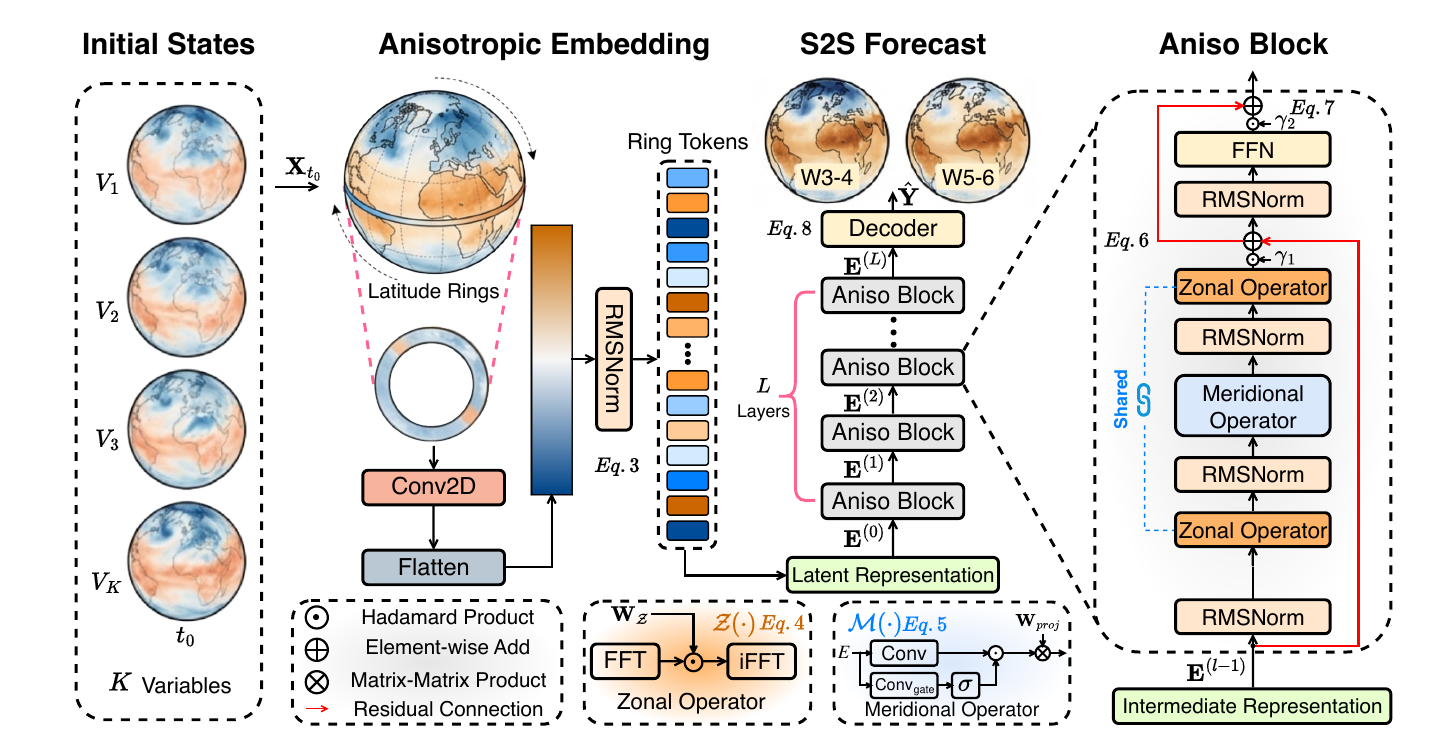}
\vspace{-2mm}
\caption{Main architecture of the proposed AnisoCast. The Anisotropic Embedding projects atmospheric initial states into latitudinal ring tokens by compressing the longitudinal dimension into feature channels. These tokens are then processed through a backbone of several stacked Aniso Blocks. Specifically, each block sequentially applies two weight-shared Zonal Operators and a Meridional Operator to separately process zonal- and meridional-oriented latent interactions. Finally, a Decoder maps the evolved features to bi-weekly forecast fields.}
\label{fig:framework}
\vspace{-2mm}
\end{figure*}

\section{Methodology}
\label{sec:methodology}
\vspace{-2mm}

Figure~\ref{fig:framework} illustrates the architecture of the proposed AnisoCast, which is designed to align neural operations with the anisotropic nature of atmospheric dynamics. It comprises three modules:
(i) \textbf{Anisotropic Embedding}, which projects each latitude ring
into a latent token by encoding longitudinal information
within feature channels;
(ii) A stack of \textbf{Aniso Blocks}, which use a symmetric latent composition inspired by Strang splitting. It employs a \textbf{Zonal Operator} to capture
the global spectral wave dynamics, and a \textbf{Meridional Operator} to model local spatial transport; (iii) A lightweight \textbf{Decoder}, which projects the encoded latent representations back to the physical space to generate the forecasts. The entire architecture is trained end-to-end by minimizing the prediction error over the S2S timescale.
Appendix~\ref{app:paradigm} shows a side-by-side
comparison of
AnisoCast and conventional isotropic atmospheric modeling.

\subsection{Anisotropic Embedding}
\label{subsec:embedding}

To tokenize global gridded fields,
recent data-driven S2S models adopt various embedding strategies
such as spherical harmonic embedding~\citep{telepit}, and Fourier embedding~\citep{TianQuans2s}.
Inspired by the circular patching of CirT~\citep{Cirt},
we adopt Anisotropic Embedding to encode longitude within
token features while retaining latitude as the sequence dimension.
Instead of extracting local patches, we project the entire latitudinal ring at each latitude index $h$ into a single high-dimensional latent token. Formally, given an input state
$\bm{X}_{t_0} \in \mathbb{R}^{K \times H \times W}$
at time $t_0$,
we apply 2D convolution
$\mathrm{Conv}_{1 \times W}$
with kernel size $(1, W)$
over the latitude and longitude
dimensions.
This compresses the longitude dimension to one spatial position, while its $C$ output filters project the $K$ input channels to $C$ latent channels.
It is then followed by flattening
$\mathrm{Flat}(\cdot)$,
which reshapes the spatial dimension to a sequence of length $H$,
and
Root Mean Square Normalization (RMSNorm) \citep{zhang2019rootmeansquarelayer}
$\mathrm{RMSN}(\cdot)$,
which stabilizes the latent features without recentering:
\begin{equation}
\label{eq:embedding}
\bm{E}^{(0)} = \mathrm{RMSN}\left( \mathrm{Flat}\left( \mathrm{Conv}_{1 \times W}(\bm{X}_{t_0}) \right) \right) \in \mathbb{R}^{H \times C}.
\end{equation}
Here,
dimension $H$
corresponds to the meridional domain, while the feature dimension $C$ encodes longitudinal information.
We call each latitude-indexed row in $\bm{E}^{(0)}$ a \emph{ring token}.
Unlike local patching
that fragments the grid~\citep{pathak2022fourcastnet, TianQuans2s}, this aggregation
constructs
each ring token
from the complete longitudinal context.
Each ring token aggregates the complete longitude span, including
both sides of the longitude-coordinate boundary, providing atmospheric context of the full latitudinal ring.

\subsection{Aniso Block}
\label{subsec:AnisoCast_block}
The ring tokens
are processed through a backbone of $L$ stacked Aniso Blocks through direction-aware operations that separately model zonal- and meridional- interactions.
This provides a novel way to capture the distinct
atmospheric behavior along the two directions.
Specifically, we define two
operators acting in this latent space:
Zonal Operator
and
Meridional Operator.

Since Anisotropic Embedding has encoded the longitudinal information
into $C$ latent channels per ring token,
the Zonal Operator
performs
spectral mixing
across these channels, enabling amplitude and phase modulation of latent
spectral components to model dependencies across the full longitude span.
Let $\bm{E}^{(l-1)} \in \mathbb{R}^{H \times C}$ be the input to the $l$-th block.
We define the Zonal Operator as:
\begin{equation}
\label{eq:zonal_operator}
\mathcal{Z}(\bm{E}^{(l-1)})
=
\mathrm{iFFT}
(
\bm{W}_{\mathcal{Z}}
\odot
\mathrm{FFT}(\bm{E}^{(l-1)})
),
\end{equation}
where
$\mathrm{FFT}(\cdot)$ and $\mathrm{iFFT}(\cdot)$ denote the one-dimensional Fast Fourier Transform and inverse Fast Fourier Transform applied along the latent channel dimension $C$ of each ring token,
respectively, and
$\odot$ is the element-wise Hadamard product.
The parameter $\bm{W}_{\mathcal{Z}}\in\mathbb{C}^{C}$ is a learnable complex-valued weight initialized near identity.
This latent spectral mixing
enables each output channel to depend on
the ring token's
input channels
in $\mathcal{O}(C\log C)$ time.

On the other hand, the Meridional Operator acts along the latitude dimension $H$,
modeling local spatial transport as gated exchanges
of latent information among neighboring latitude rings:
\begin{equation}
\label{eq:meridional_operator}
\mathcal{M}(\bm{E}^{(l-1)})
=
(
\mathrm{Conv}(\bm{E}^{(l-1)})
\odot
\sigma(\mathrm{Conv}_{\mathrm{gate}}(\bm{E}^{(l-1)}))
)
\bm{W}_{\mathrm{proj}},
\end{equation}
where $\sigma(\cdot)$ is the sigmoid gate, $\mathrm{Conv}(\cdot)$ and
$\mathrm{Conv}_{\mathrm{gate}}(\cdot)$ are two 1D convolution layers with kernel size $k$ sliding over the meridional sequence dimension $H$, and $\bm{W}_{\mathrm{proj}}$ is a learnable projection matrix.
The gated structure adaptively controls information flow along the meridional direction.
Together, the two operators avoid pairwise attention across
latitude rings, so their mixing costs scale linearly with
the number of rings $H$.

Rather than treating the block as a physical time integrator on the primitive atmospheric variables, we adopt the symmetric composition underlying Strang splitting~\citep{strang1968construction,hairer2006geometric}
in the latent space. Specifically,
the block updates the representation as:
\begin{equation}
    \label{eq:symmetric_mixing}
    \bm{E}'
    =
    \bm{E}^{(l-1)}
    +
    \bm{\gamma}_1
    \odot
    \left(
    \mathcal{Z}
    \circ \mathrm{RMSN}
    \circ \mathcal{M}
    \circ \mathrm{RMSN}
    \circ \mathcal{Z}
    \right)
    \left(
    \mathrm{RMSN}(\bm{E}^{(l-1)})
    \right),
\end{equation}
followed by a Feedforward Network ($\mathrm{FFN}(\cdot)$) with residual connections:
\begin{equation}
    \label{eq:ffn_residual}
    \bm{E}^{(l)}
    =
    \bm{E}'
    +
    \bm{\gamma}_2
    \odot
    \mathrm{FFN}(\mathrm{RMSN}(\bm{E}')).
\end{equation}
Here, $\circ$ denotes operator composition, and
$\bm{\gamma}_1,\bm{\gamma}_2 \in \mathbb{R}^{C}$ are learnable LayerScale~\citep{layerscale} parameters that modulate the residual branches and stabilize training.
In Eq.~(\ref{eq:symmetric_mixing}), the Zonal Operator is applied twice with shared weights, yielding a symmetric $\mathcal{Z} \circ \mathcal{M} \circ \mathcal{Z}$ composition.\footnote{We also evaluate the variant $\mathcal{M}\circ\mathcal{Z}\circ\mathcal{M}$ with shared outer Meridional Operators in Appendix~\ref{app:full_ablation}.}
This also helps to mitigate the ordering bias, as the zonal wave
propagation and meridional transport processes are
non-commutative (Section~\ref{subsec:operator_splitting}).

To motivate the symmetric ordering, we apply the classical
Strang-splitting result from Section~\ref{subsec:operator_splitting}
to a linear latent reference system.
Let
$\bm{z}_0=\operatorname{vec}(\mathrm{RMSN}(\bm{E}^{(l-1)}))
\in\mathbb{R}^{HC}$,
where $\operatorname{vec}(\cdot)$ stacks the entries into a vector,
and let
$\mathbf{A}_{\mathrm{lat}},\mathbf{B}_{\mathrm{lat}}
\in\mathbb{R}^{HC\times HC}$
be fixed matrices defining its two components.
For a composition step size $\tau>0$, Eq.~(\ref{eq:strang}) gives
\(
\mathcal{S}^{\mathrm{lat}}(\tau)
=
\exp\!\left(\frac{\tau}{2}\mathbf{A}_{\mathrm{lat}}\right)
\exp\!\left(\tau\mathbf{B}_{\mathrm{lat}}\right)
\exp\!\left(\frac{\tau}{2}\mathbf{A}_{\mathrm{lat}}\right).
\)
The classical Strang-splitting analysis
\citep{strang1968construction,hairer2006geometric}
establishes an $\mathcal{O}(\tau^3)$ local error bound for
this linear composition relative to
$\exp(\tau(\mathbf{A}_{\mathrm{lat}}+\mathbf{B}_{\mathrm{lat}}))$
as $\tau\to0$.
This provides motivation for the
$\mathcal{Z}\circ\mathcal{M}\circ\mathcal{Z}$ ordering in
Eq.~(\ref{eq:symmetric_mixing}).

While Layer Normalization (LayerNorm)~\citep{ba2016layer}
has been widely adopted in S2S architectures~\citep{Cirt,EIMP,telepit},
its mean-centering discards each ring token's channel mean,
which may encode background-state information relevant to S2S forecasting.
We therefore adopt RMSNorm, whose normalization core rescales features
without recentering.
Appendix~\ref{appendix:proof_norm} provides a
qualitative analysis showing that
the LayerNorm core removes the zero-frequency component
of the latent channel spectrum, whereas the RMSNorm core
retains all spectral coefficients up to a common positive
scaling factor, thereby avoiding this loss of channel-mean information.
Ablation results in Section~\ref{subsec:ablation}
further show that this choice improves forecasting accuracy.

\subsection{Decoder}
\label{subsec:Decoder}
\vspace{-2mm}
After passing through $L$ Aniso Blocks, the latent representation $\bm{E}^{(L)} \in \mathbb{R}^{H \times C}$ contains the encoded
representation extracted from the initial
atmospheric state $\bm{X}_{t_0}$. The Decoder projects this representation back to the original physical grid to generate the final forecasts. We employ a linear projection layer to reconstruct the longitudinal dimension $W$ from the
$C$
feature channels:
\begin{equation}
    (\hat{\bm{Y}}_{15:28},\, \hat{\bm{Y}}_{29:42}) = \mathrm{Reshape}\left( \mathrm{Linear}(\bm{E}^{(L)}) \right),
\end{equation}
where the linear layer projects the feature dimension from $C$ to $2 \times K \times W$.
Finally, we obtain two tensors, $\hat{\bm{Y}}_{15:28}$ and $\hat{\bm{Y}}_{29:42}$, each with dimension
$\mathbb{R}^{K \times H \times W}$, representing the two target bi-weekly mean fields corresponding to weeks 3--4 and weeks 5--6, respectively.

\noindent\textbf{Training Objective.}
The model is trained in an end-to-end manner, minimizing the
average squared error
of the predicted bi-weekly means
over the two forecast windows
as in Eq.~(\ref{eq:problem_def}):
\begin{equation}
\mathcal{L} \!=\! \frac{1}{2 K H W}\! \left(
\left\| \hat{\bm{Y}}_{15:28}\! - \!\bm{Y}_{15:28} \right\|_{\mathcal{L}_2}^2
\!\!\!+\!\!\left\| \hat{\bm{Y}}_{29:42} \!-\! \bm{Y}_{29:42} \right\|_{\mathcal{L}_2}^2
\right)\!,
\end{equation}
where $\bm{Y}_{15:28}$ and $\bm{Y}_{29:42}$ are the ground-truth averaged fields for weeks 3--4 and weeks 5--6, respectively, and
$\|\cdot\|_{\mathcal{L}_2}^2$ is the latitude-weighted Mean Squared Error
\citep{nguyen2023climax,STORM,TianQuans2s}
that
accounts for the
Earth's
spherical geometry
by weighting
the error at latitude $h$
with $w_h = \cos(\varphi_h)$.

\subsection{Computational Complexity }
\label{subsec:complexity}
\vspace{-2mm}

We analyze the computational complexity of AnisoCast in terms of the latitude resolution $H$, longitude resolution $W$, number of variables $K$, latent dimension $C$, number of blocks $L$, and meridional kernel size $k$. The full complexity of AnisoCast is $\mathcal{O}(HKWC)+\mathcal{O}(LH(C^2+C\log C+Ck))+\mathcal{O}(HCKW)$, corresponding to Anisotropic Embedding, the stacked Aniso Block, and the Decoder, respectively.
Thus, the full model still scales linearly with the longitudinal resolution $W$, since the embedding must read the original longitude-resolved inputs and the decoder must reconstruct longitude-resolved forecasts.
CirT~\citep{Cirt} also uses $H$ ring tokens,
but its global attention incurs an $\mathcal{O}(H^2C)$
token-mixing cost per block.
ClimaX~\citep{nguyen2023climax} applies global attention
over $N$ spatial patch tokens, with an $\mathcal{O}(N^2C)$
token-mixing cost per block, where $N\propto HW$ for a
fixed patch size.
AnisoCast avoids these quadratic attention terms through
latent spectral mixing and local meridional convolutions,
giving a backbone cost that scales linearly with $H$.
Detailed analysis are provided in Appendix~\ref{app:model_complexity}.

\vspace{-1mm}
\section{Experiments}
\vspace{-2mm}
\subsection{Experimental Setup}
\label{subsec:experimental_setup}
\vspace{-2mm}
\noindent\textbf{Dataset.}
Following the established benchmarking protocols~\citep{Cirt,EIMP,telepit}, we evaluate AnisoCast on the ERA5 reanalysis dataset~\citep{hersbach2020era5}. We regrid the data to a $1.5^{\circ}$ resolution, resulting in a $121 \times 240$ latitude–longitude grid. Following \citep{Cirt,EIMP},
we use 63 meteorological variables as inputs at
initialization and as prediction targets averaged over each future
14-day window,
including six pressure-level variables (geopotential $z$, specific humidity $q$, temperature $t$, u-wind $u$, v-wind $v$, and vertical velocity $w$) at ten pressure levels (10, 50, 100, 200, 300, 500, 700, 850, 925, and 1000~hPa), and three surface variables (2-meter temperature $t2m$, 10-meter u-wind $u10$, and 10-meter v-wind $v10$).
Following \citet{Cirt,EIMP}, we split the dataset chronologically, using the 1979–2016 data (38 years) for training, the 2017 data for validation, and the 2018 data for testing. Details can be found in Appendix \ref{app:datasets}.

\noindent\textbf{Baselines.}
We compare AnisoCast with two categories of methods:
(1) \textbf{Operational NWP systems}: CMA~\citep{wu2019beijing},
NCEP~\citep{saha2014ncep}, and ECMWF~\citep{molteni1996ecmwf};
(2) \textbf{Data-driven models}: PanguWeather
(Pangu)~\citep{bi2023accurate}, FourCastNetV2
(FCN2)~\citep{pathak2022fourcastnet},
ClimaX~\citep{nguyen2023climax},
ClimODE~\citep{verma2024climode},
CirT~\citep{Cirt}, EIMP~\citep{EIMP},
TelePiT~\citep{telepit},
TianQuan-S2S~\citep{TianQuans2s}, and
OmniCast~\citep{nguyen2025omnicast}.
We align predictions to the same bi-weekly targets and verification
grid. More details are provided in Appendix~\ref{app:baselines}.

\noindent\textbf{Metrics.}
Following existing evaluation standards~\citep{rasp2020weatherbench,chaosbench}, we employ the Latitude-weighted Root Mean Square Error (RMSE) and Anomaly Correlation Coefficient (ACC)
to quantify performance~\citep{chaosbench}, defined as:
\begin{equation}
\label{eq:metrics}
\scalebox{0.85}{$
\begin{alignedat}{2}
\text{RMSE}
&= \sqrt{\frac{1}{H W} \sum_{h, w} \alpha(h)\,
(\hat{\bm{X}}_{h,w} - \bm{X}_{h,w})^2},
\qquad
&
\text{ACC}
&=
\frac{
\sum_{h,w} \alpha(h)\, \hat{\bm{X}}'_{h, w}\, \bm{X}'_{h, w}
}{
\sqrt{\sum_{h,w} \alpha(h)\, (\hat{\bm{X}}'_{h, w})^2}\;
\sqrt{\sum_{h,w} \alpha(h)\, (\bm{X}'_{h, w})^2}
},
\end{alignedat}
$}
\end{equation}
where $\alpha(h) = \cos(\phi_{h}) \big/ \left(\frac{1}{H} \sum_{h'} \cos(\phi_{h'})\right)$ is the latitude weighting factor, $\bm{X}$ and $\hat{\bm{X}}$ are the ground truth and prediction for a given bi-weekly verification window. The anomalies are computed as $\bm{X}' = \bm{X} - \bm{C}$ and $\hat{\bm{X}}' = \hat{\bm{X}} - \bm{C}$, where $\bm{C}$ is the observational climatology.

\noindent\textbf{Implementation Details.}
For PanguWeather~\citep{bi2023accurate} and FourCastNetV2~\citep{pathak2022fourcastnet}, we utilize the official pre-trained weights provided by the ECMWF project through API (see Appendix \ref{app:baselines} for details).
For the other data-driven models, we train
them from scratch using their officially recommended configurations on the same GPU server. Training is conducted on 8 NVIDIA H20 80G GPUs.
Inference is
performed
on a single NVIDIA H100 80G GPU.
We use $L=7$ blocks for the proposed method.
More implementation details are in Appendix~\ref{app:implementation}.

\begin{table*}[t]
\centering
\footnotesize
\setlength{\tabcolsep}{3pt}
\renewcommand{\arraystretch}{1.05}
\caption{Forecasting results for 14-day mean targets
at weeks~3--4 and weeks~5--6. The best result is in \textbf{bold},
and the second-best is \underline{underlined}. Results for the operational NWP systems are obtained from~\citep{telepit}, whereas we reproduce other data-driven models and report their results. Pressure-level variables $z500$~(gpm), $z850$~(gpm) and $t500$~(K), $t850$~(K) are at 500 and 850~hPa, while single-level variables $t2m$~(K), $u10$~(m/s), and $v10$~(m/s) are measured at 2m and 10m above the surface, respectively. N/A indicates the baseline does not predict the variables~\citep{EIMP}.}
\label{tab:mainresults}
\vspace{-2mm}
\begin{threeparttable}
\scriptsize
\begin{tabular*}{\linewidth}{@{\extracolsep{\fill}}c cc | *{4}{c} | *{3}{c} | *{4}{c} | *{3}{c}}
\toprule
  & \multicolumn{2}{c|}{\multirow{3}{*}{\textbf{Model}}}
  & \multicolumn{7}{c|}{\textbf{RMSE} ($\downarrow$)}
  & \multicolumn{7}{c}{\textbf{ACC} ($\uparrow$)} \\
\cmidrule(lr){4-10}\cmidrule(lr){11-17}
  & &
  & \multicolumn{4}{c|}{\textbf{Pressure Level}}
  & \multicolumn{3}{c|}{\textbf{Single Level}}
  & \multicolumn{4}{c|}{\textbf{Pressure Level}}
  & \multicolumn{3}{c}{\textbf{Single Level}} \\
\cmidrule(lr){4-7}\cmidrule(lr){8-10}
\cmidrule(lr){11-14}\cmidrule(lr){15-17}
  & &
  & $z500$ & $z850$ & $t500$ & $t850$
  & $t2m$ & $u10$ & $v10$
  & $z500$ & $z850$ & $t500$ & $t850$
  & $t2m$ & $u10$ & $v10$ \\
\midrule
\multirow{13}{*}{\rotatebox{90}{\textbf{Weeks 3--4 (Days 15--28)}}}
  & \multicolumn{2}{c|}{CMA}
  & 69.833 & 47.056 & 2.598 & 2.918 & 4.221 & N/A & N/A
  & 0.968 & 0.919 & 0.975 & 0.974 & 0.960 & N/A & N/A \\
  & \multicolumn{2}{c|}{NCEP}
  & 68.441 & 44.787 & 2.225 & 2.742 & 4.028 & N/A & N/A
  & 0.969 & 0.921 & 0.979 & 0.977 & 0.966 & N/A & N/A \\
  & \multicolumn{2}{c|}{ECMWF}
  & 63.674 & 42.349 & 2.067 & 2.391 & \textbf{3.232} & N/A & N/A
  & 0.972 & 0.929 & 0.982 & 0.983 & 0.971 & N/A & N/A \\
\noalign{\vskip 1pt}
\cdashline{2-17}[1.2pt/2.5pt]
\noalign{\vskip 1pt}
  & \multicolumn{2}{c|}{Pangu}
  & 66.224 & 37.959 & 2.271 & 2.569 & N/A & 2.431 & 1.984
  & 0.963 & 0.926 & 0.966 & 0.963 & N/A & 0.812 & 0.686 \\
  & \multicolumn{2}{c|}{FCN2}
  & 62.755 & 41.020 & 2.093 & 2.390 & N/A & 2.328 & 1.896
  & 0.947 & 0.896 & 0.966 & 0.957 & N/A & 0.830 & 0.712 \\
\noalign{\vskip 1pt}
\cdashline{2-17}[1.2pt/2.5pt]
\noalign{\vskip 1pt}
  & \multicolumn{2}{c|}{ClimaX}
  & 59.064 & 38.420 & 2.142 & 3.242 & 5.342 & 2.360 & 2.025
  & 0.978 & 0.945 & 0.981 & 0.971 & 0.956 & 0.807 & 0.649 \\
  & \multicolumn{2}{c|}{ClimODE}
  & 68.767 & 39.326 & 2.518 & 3.147 & 3.782 & 2.326 & 1.793
  & 0.970 & 0.946 & 0.972 & 0.970 & \underline{0.975} & 0.827 & 0.738 \\
  & \multicolumn{2}{c|}{CirT}
  & 55.505 & \underline{34.345} & 1.931 & 2.265 & 5.478 & 2.253 & 1.873
  & 0.980 & 0.955 & \underline{0.984} & 0.983 & 0.958 & 0.835 & 0.716 \\
  & \multicolumn{2}{c|}{EIMP}
  & 66.498 & 41.820 & 2.312 & 2.800 & 4.785 & 3.683 & 2.137
  & \underline{0.981} & \underline{0.956} & 0.983 & \underline{0.984}
  & 0.974 & 0.672 & 0.654 \\
  & \multicolumn{2}{c|}{TelePiT}
  & 59.772 & 35.199 & 2.035 & 2.448 & 4.002 & 2.122 & \underline{1.746}
  & 0.977 & 0.953 & 0.982 & 0.980 & 0.974 & \underline{0.854} & 0.741 \\
  & \multicolumn{2}{c|}{TianQuan-S2S}
  & \underline{53.212} & 34.662 & \underline{1.898} & \underline{2.232}
  & 3.286 & \underline{2.022} & 1.885
  & 0.973 & 0.955 & 0.982 & 0.975 & 0.966 & 0.843 & \underline{0.768} \\
  & \multicolumn{2}{c|}{OmniCast}
  & 54.423 & 35.561 & 1.992 & 2.341 & 3.366 & 2.213 & 1.903
  & 0.972 & 0.955 & 0.980 & 0.974 & 0.965 & 0.843 & 0.766 \\
  & \multicolumn{2}{c|}{\textbf{AnisoCast}}
  & \textbf{51.956} & \textbf{33.453} & \textbf{1.824} & \textbf{2.145}
  & \underline{3.255} & \textbf{1.994} & \textbf{1.621}
  & \textbf{0.982} & \textbf{0.958} & \textbf{0.985} & \textbf{0.985}
  & \textbf{0.982} & \textbf{0.872} & \textbf{0.783} \\
\midrule
\multirow{13}{*}{\rotatebox{90}{\textbf{Weeks 5--6 (Days 29--42)}}}
  & \multicolumn{2}{c|}{CMA}
  & 71.400 & 48.583 & 2.639 & 2.933 & 5.521 & N/A & N/A
  & 0.966 & 0.912 & 0.974 & 0.974 & 0.950 & N/A & N/A \\
  & \multicolumn{2}{c|}{NCEP}
  & 71.332 & 46.834 & 2.343 & 2.833 & 4.406 & N/A & N/A
  & 0.967 & 0.913 & 0.977 & 0.976 & 0.952 & N/A & N/A \\
  & \multicolumn{2}{c|}{ECMWF}
  & 68.116 & 45.330 & 2.177 & 2.511 & 4.080 & N/A & N/A
  & 0.968 & 0.916 & 0.979 & 0.981 & 0.970 & N/A & N/A \\
\noalign{\vskip 1pt}
\cdashline{2-17}[1.2pt/2.5pt]
\noalign{\vskip 1pt}
  & \multicolumn{2}{c|}{Pangu}
  & 76.939 & 47.041 & 2.829 & 2.998 & N/A & 2.679 & 2.104
  & 0.956 & 0.911 & 0.958 & 0.957 & N/A & 0.783 & 0.655 \\
  & \multicolumn{2}{c|}{FCN2}
  & 66.531 & 43.469 & 2.250 & 2.567 & N/A & 2.479 & 1.980
  & 0.943 & 0.889 & 0.963 & 0.953 & N/A & 0.812 & 0.691 \\
\noalign{\vskip 1pt}
\cdashline{2-17}[1.2pt/2.5pt]
\noalign{\vskip 1pt}
  & \multicolumn{2}{c|}{ClimaX}
  & 58.712 & 38.271 & 2.141 & 3.300 & 5.354 & 2.394 & 2.031
  & 0.978 & 0.945 & 0.982 & 0.971 & 0.956 & 0.800 & 0.643 \\
  & \multicolumn{2}{c|}{ClimODE}
  & 68.218 & 39.287 & 2.554 & 3.213 & 4.094 & 2.336 & 1.844
  & 0.971 & 0.946 & 0.972 & 0.969 & 0.971 & 0.829 & 0.717 \\
  & \multicolumn{2}{c|}{CirT}
  & 55.214 & \underline{33.940} & \underline{1.920} & \underline{2.265}
  & 5.652 & 2.251 & 1.826
  & 0.980 & 0.956 & \underline{0.984} & \underline{0.983}
  & 0.954 & 0.835 & 0.726 \\
  & \multicolumn{2}{c|}{EIMP}
  & 66.856 & 41.926 & 2.329 & 2.858 & 5.018 & 2.555 & 2.155
  & \underline{0.981} & \underline{0.957} & 0.983 & \underline{0.983}
  & 0.972 & 0.792 & 0.651 \\
  & \multicolumn{2}{c|}{TelePiT}
  & 58.770 & 35.164 & 1.953 & 2.451 & 4.056
  & \underline{2.120} & \underline{1.750}
  & 0.978 & 0.953 & 0.982 & 0.981 & \underline{0.975}
  & \underline{0.854} & 0.741 \\
  & \multicolumn{2}{c|}{TianQuan-S2S}
  & \underline{54.623} & 35.551 & 1.991 & 2.368 & 3.446 & 2.121 & 1.908
  & 0.974 & 0.955 & 0.983 & 0.978 & 0.970 & 0.845 & \underline{0.769} \\
  & \multicolumn{2}{c|}{OmniCast}
  & 55.522 & 36.015 & 1.991 & 2.340 & \underline{3.347} & 2.326 & 1.912
  & 0.973 & 0.956 & 0.980 & 0.973 & 0.964 & 0.841 & 0.765 \\
  & \multicolumn{2}{c|}{\textbf{AnisoCast}}
  & \textbf{52.268} & \textbf{33.437} & \textbf{1.830} & \textbf{2.156}
  & \textbf{3.231} & \textbf{1.993} & \textbf{1.606}
  & \textbf{0.982} & \textbf{0.958} & \textbf{0.985} & \textbf{0.985}
  & \textbf{0.982} & \textbf{0.873} & \textbf{0.786} \\
\bottomrule
\end{tabular*}
\end{threeparttable}
\vspace{-1mm}
\end{table*}

\subsection{Overall Performance}
\vspace{-2mm}
Table~\ref{tab:mainresults}
reports the quantitative comparison with operational NWP systems and data-driven baselines. AnisoCast often achieves the best performance,
with particularly stable gains on large-scale upper-air variables ($z500$, $z850$, $t500$, and $t850$)
and competitive performance on near-surface variables ($t2m$, $u10$, and $v10$).
This suggests that the anisotropic ring-token representation and symmetric latent zonal--meridional composition provide an effective inductive bias for direct S2S forecasting.

Additional results are reported in the appendices. Appendix~\ref{app:seed} show that
RMSE and ACC
have small standard deviations
across five seeds, indicating
robustness
to random initialization
across training runs.
Appendix~\ref{app:extended-period} extends testing from 2018
to 2018--2019.
AnisoCast achieves the lowest RMSE on all seven reported variables and the highest ACC on six.
Appendix~\ref{app:monthlyvisual} shows that AnisoCast remains
competitive across months, particularly on $t500$ and $t850$,
indicating that its gains are not confined to a single season.
Appendix~\ref{app:additional_criteria} reports higher temporal
correlations for $z500$ and $z850$ in both windows and better
extreme-event detection scores for $u10$, $v10$, and $t2m$,
extending the evaluation beyond spatial RMSE and ACC.

\begin{figure}[!t]
    \centering
    \includegraphics[width=\linewidth]{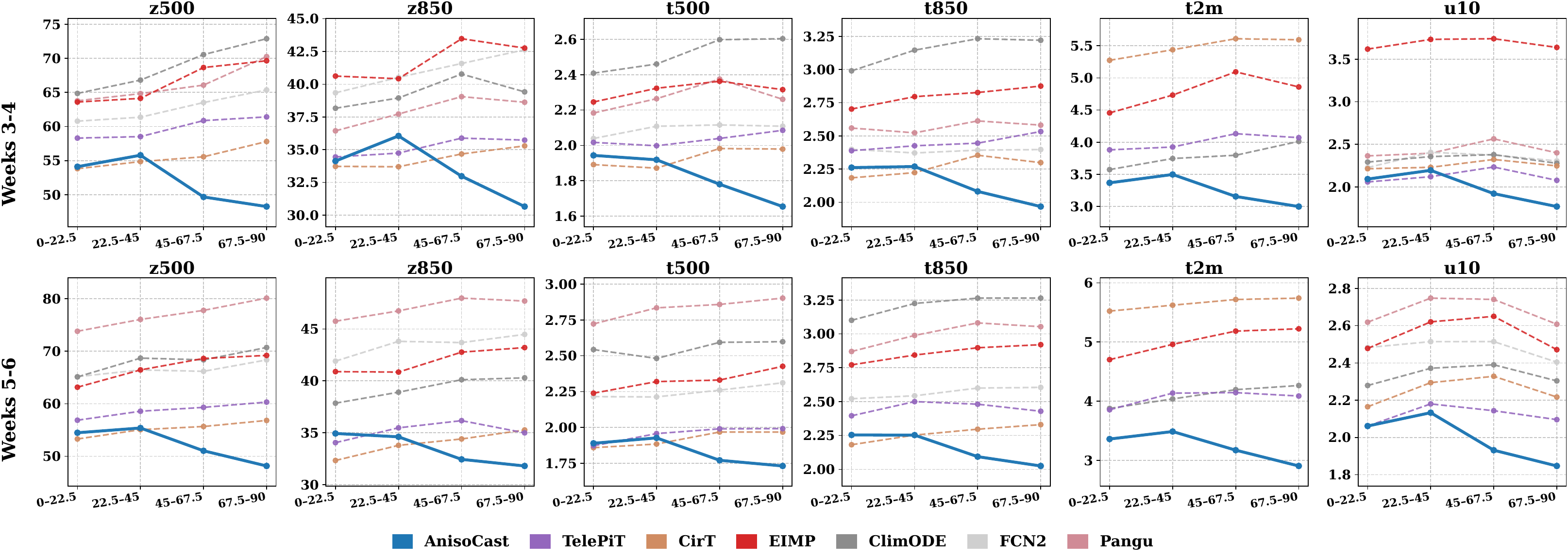}
    \vspace{-6mm}
\caption{RMSE variation w.r.t. different latitude bands, comparing AnisoCast with baseline models.}
    \label{fig:lat_rmse}
    \vspace{-1em}
\end{figure}

\subsection{Additional Model Analysis}

\vspace{-2mm}
For model analysis, we select $z$ and  $t$
at multiple pressure levels, including $z500$ and $t850$, together with
near-surface temperature ($t2m$), following prior S2S studies~\citep{Cirt,EIMP}.
These variables provide complementary information on atmospheric conditions:
$z500$ characterizes synoptic-scale circulation and $t850$ reflects broad
temperature variations~\citep{rasp2020weatherbench}, while $t2m$ represents
near-surface thermal conditions.

\noindent\textbf{Quantitative Latitudinal Analysis.}
Figure~\ref{fig:lat_rmse} reports RMSE across four latitudinal bands from the equator ($0^\circ$) to the pole ($90^\circ$).
AnisoCast shows smaller error growth toward high latitudes than competing baselines under the same ERA5-based verification protocol, suggesting that the anisotropic ring embedding reduces sensitivity to latitude-dependent grid convergence and projection distortion.
Appendix~\ref{app:resolution} reports lower $z500$ and $t2m$
RMSE under three grid resolutions
($5.625^\circ$, $1.5^\circ$, and $0.25^\circ$).
Appendix~\ref{app:non_polar_eval} also reports lower RMSE for
$z500$, $z850$, $u10$, and $v10$ after excluding regions
beyond $60^\circ$ latitude in either hemisphere.
These results indicate that the gains of our proposed model are not confined to
one resolution or to high-latitude regions.

\begin{figure}[!t]
    \centering
    \includegraphics[width=\linewidth]{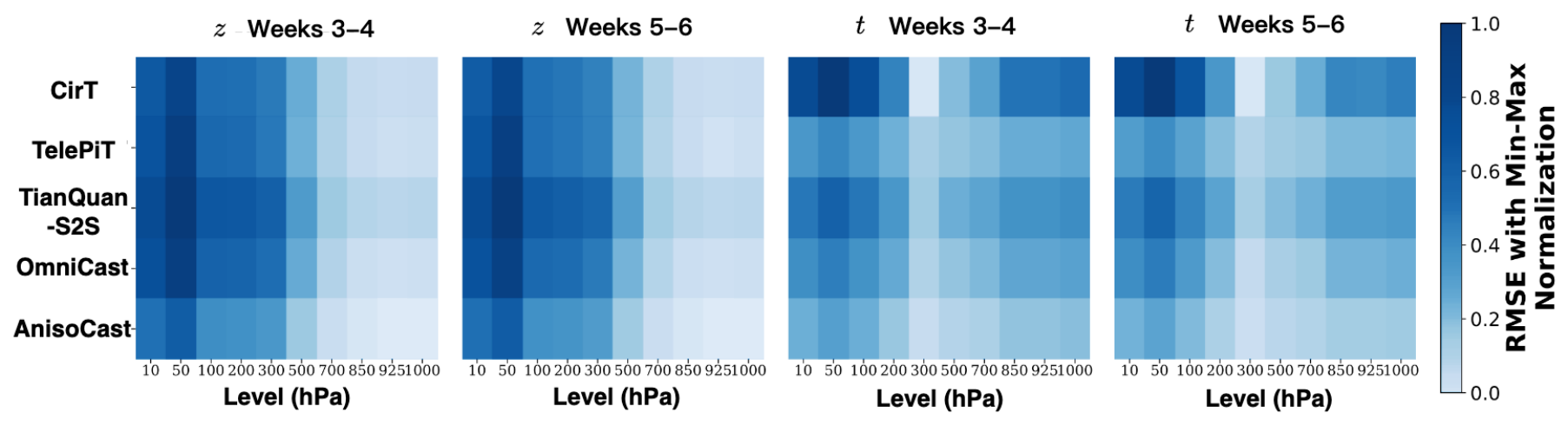}
    \vspace{-7mm}
    \caption{RMSE comparison on variable $z$ and $t$ across different pressure levels.}
    \label{fig:vertical_rmse}
    \vspace{-2mm}
\end{figure}

\noindent\textbf{Performance across Pressure Levels.}
Following~\citep{Cirt,telepit}, we visualize the RMSE
for $z$ and $t$ across 10 pressure levels in Figure~\ref{fig:vertical_rmse}.
AnisoCast maintains lower errors across most evaluated levels,
including the upper atmosphere (50--200~hPa).
This consistency reflects the joint encoding of variables
across pressure levels, which allows each prediction to use
cross-level context alongside complementary zonal and meridional
interactions. Appendix~\ref{app:SpecRes} reports lower spectral residuals~\citep{telepit} across most evaluated levels, indicating closer
agreement between predicted and target power spectra.
Appendix~\ref{app:latidunal} further shows lower $z$ and $t$
RMSE than operational NWP systems~\citep{wu2019beijing,saha2014ncep,molteni1996ecmwf} across most pressure levels,
supporting the consistency of these gains across the vertical
atmospheric column.

\begin{figure}[!t]
    \centering
    \includegraphics[width=\linewidth, trim = 8 16 8 8, clip]{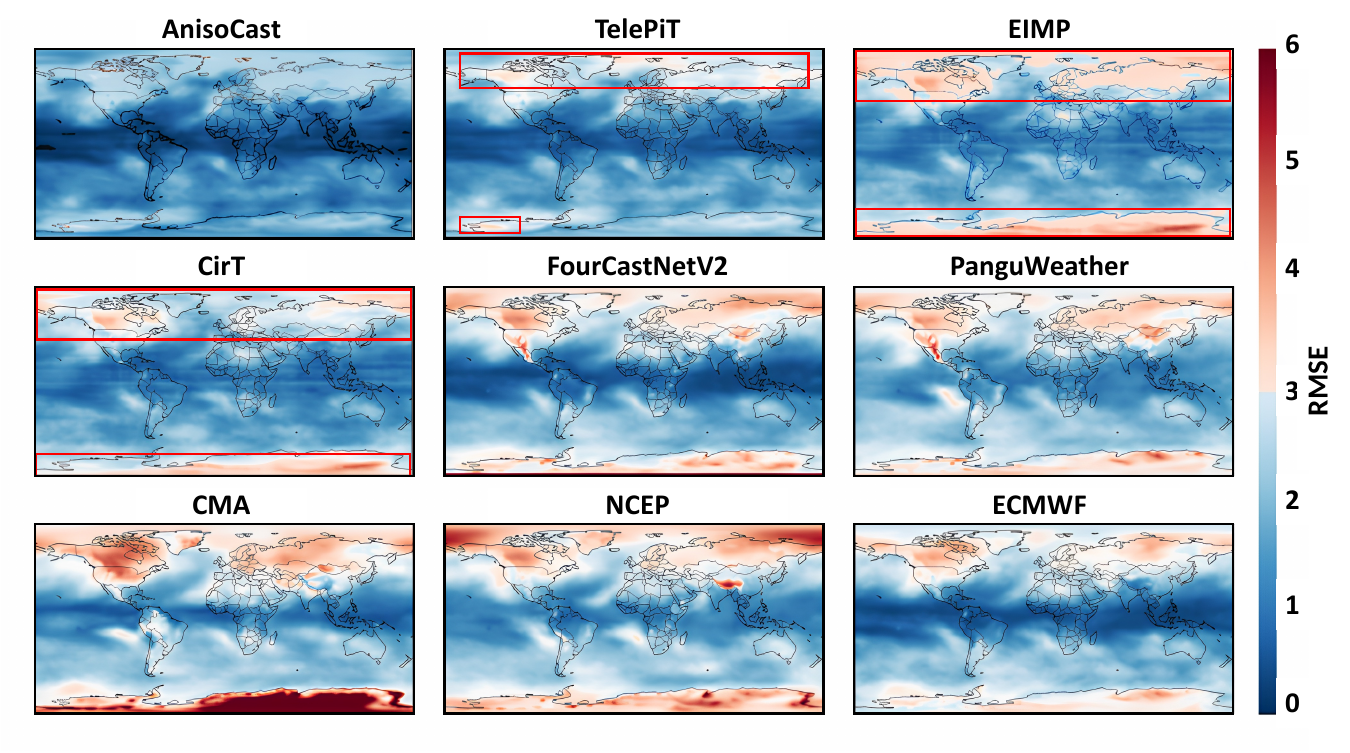}
    \vspace{-2mm}
    \caption{Global RMSE distribution of $t850$ with lead times weeks 3-4 in the testing set.}
    \label{fig:t850_w3-4}
    \vspace{-3mm}
\end{figure}

\noindent\textbf{Global Visualization.}
Figure~\ref{fig:t850_w3-4} visualizes the spatial distribution
of RMSE for $t850$. Compared with the baselines, AnisoCast shows lower spatial errors,
particularly in the high-latitude regions highlighted by the red boxes.
Retaining complete longitudinal context within each ring and
exchanging information between neighboring latitudes enables
joint modeling of broad zonal structures and local meridional
variations, consistent with these spatial improvements. Additional results in Appendix~\ref{app:dailyvisal} report slower error growth for separately trained daily-output models over days 15--42,
extending the evaluation beyond bi-weekly mean prediction. Additional results in Appendix~\ref{app:globalvis} show similar spatial improvements across additional variables and both forecast windows.

\vspace{-2mm}
\subsection{Ablation Study}
\label{subsec:ablation}
\vspace{-1mm}

\begin{figure}[!t]
    \centering
    \includegraphics[width=\textwidth, trim=0 25 0 0, clip]{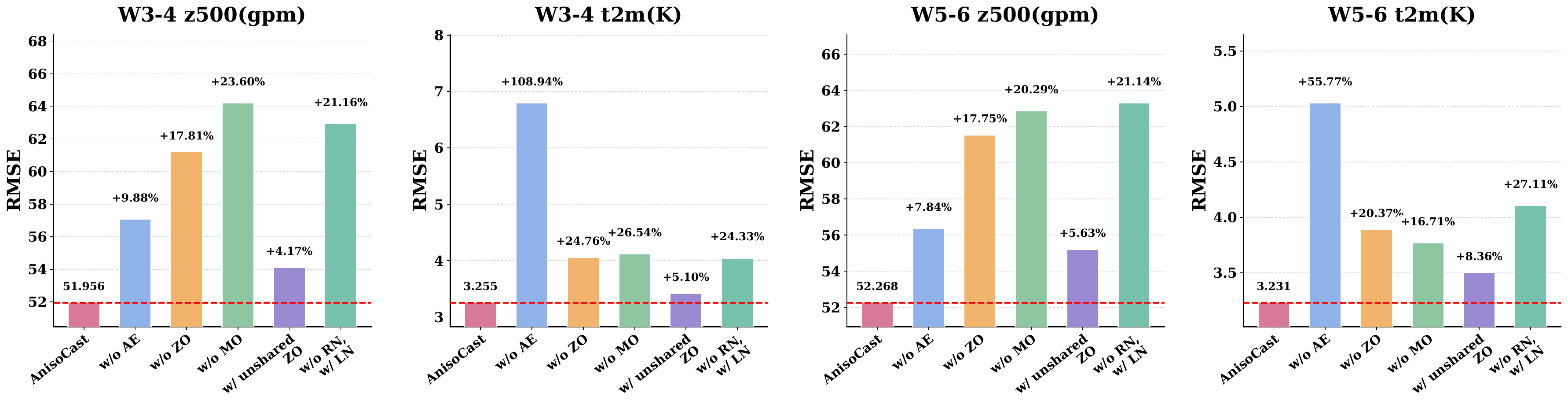}
    \vspace{-5mm}
    \caption{Ablation study quantifying the contribution of each AnisoCast component.}
    \vspace{-2mm}
    \label{fig:ablation}
\end{figure}

We conduct ablation studies to evaluate the contribution of each component of AnisoCast. Figure~\ref{fig:ablation} compares AnisoCast with
variants in which several key design choices are removed or replaced.
Here, $z500$ and $t2m$ are selected to represent upper-air circulation
and near-surface temperature, respectively.
The detailed experimental settings are in Appendix~\ref{app:ablation_settings}.
We evaluate the following variants:
(1) \textbf{w/o AE}, which replaces Anisotropic Embedding with the order-agnostic, shape-matched ring embedding;
(2)
\textbf{w/o ZO}, which replaces
$\mathcal{Z}\circ\mathcal{M}\circ\mathcal{Z}$ with
$\mathcal{M}\circ\mathcal{M}$; (3) \textbf{w/o MO}, which replaces
$\mathcal{Z}\circ\mathcal{M}\circ\mathcal{Z}$ with
$\mathcal{Z}\circ\mathcal{Z}$; (4) \textbf{unshared ZO}, which uses separate weights for the two Zonal Operators; and (5) \textbf{LayerNorm}, which replaces RMSNorm with LayerNorm.
As can be seen, all five variants increase RMSE across the evaluated variables,
supporting the contributions of the proposed embedding,
complementary zonal and meridional mixing, weight sharing,
and RMSNorm. Appendix~\ref{app:full_ablation} reports lower RMSE for
$\mathcal{Z}\circ\mathcal{M}\circ\mathcal{Z}$ than for
$\mathcal{M}\circ\mathcal{Z}\circ\mathcal{M}$, supporting
the adopted configuration.
Appendix~\ref{app:direct_ar} also shows that direct prediction
outperforms autoregressive rollout in both windows, supporting
direct optimization of the bi-weekly targets.

\vspace{-2mm}
\subsection{Efficiency Comparison}

Figure~\ref{fig:efficiency} highlights the efficiency of AnisoCast. The reported metrics are measured with batch size 1 and hidden dimension 256 for a fair comparison. AnisoCast achieves the lowest FLOPs among the compared models and delivers inference speedups of over 2$\times$ and 5$\times$ relative to CirT~\citep{Cirt} and ClimaX~\citep{nguyen2023climax}, respectively. These gains reflect its attention-free design, which replaces quadratic attention with Zonal and Meridional Operators that avoid quadratic token-mixing costs.

\vspace{-1mm}

\begin{wrapfigure}{r}{0.50\textwidth}
    \centering
    \includegraphics[width=\linewidth, trim=5 5 5 45, clip]{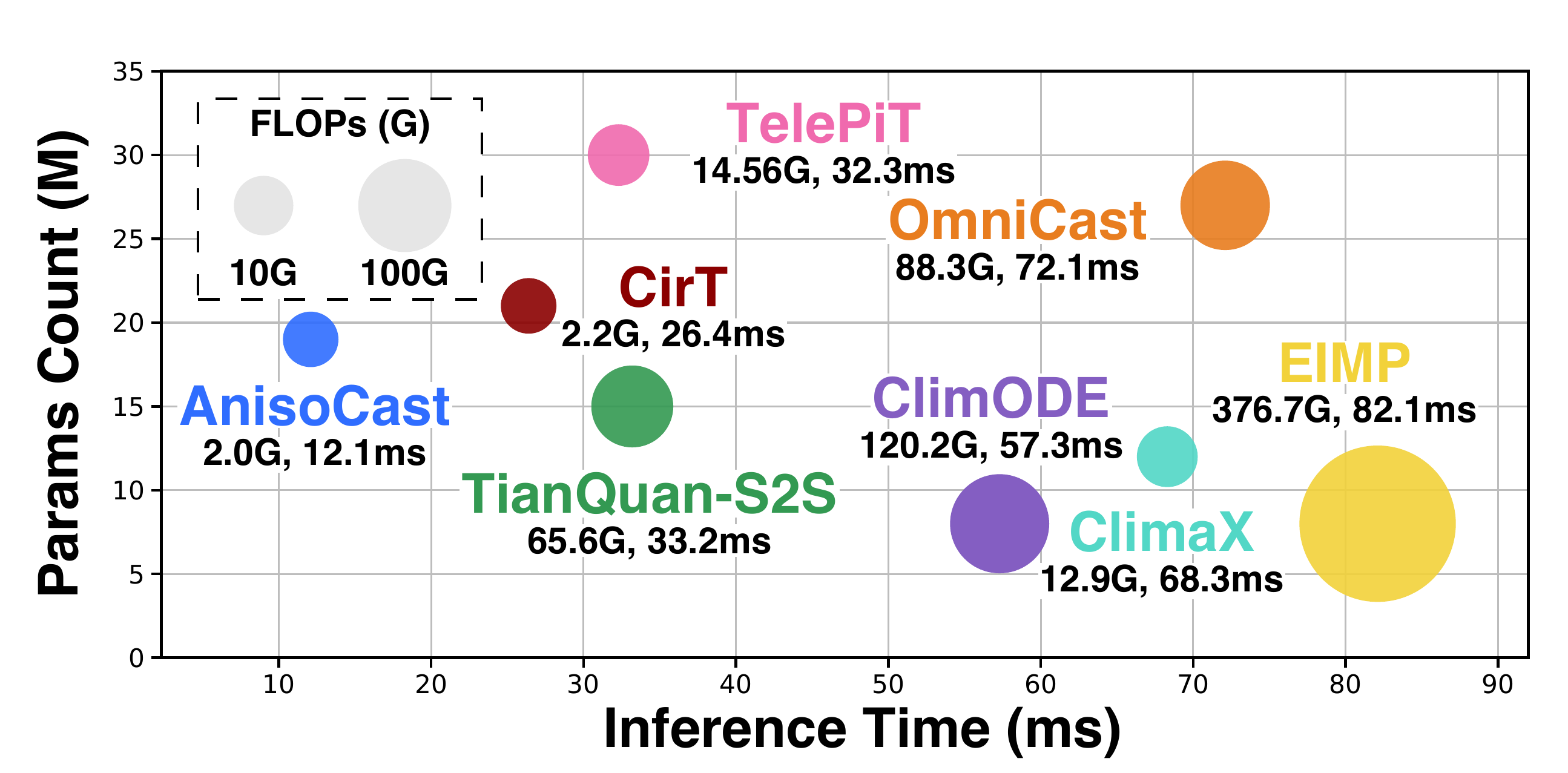}
    \vspace{-5mm}
    \caption{Efficiency analysis. The horizontal axis shows inference latency, the vertical axis shows parameter count, and bubble size depicts FLOPs.}
    \label{fig:efficiency}
    \vspace{-2mm}
\end{wrapfigure}

\vspace{-2mm}

\section{Conclusion and Future Work}
\vspace{-3mm}
This paper explores directional structure as an inductive bias for modeling anisotropic atmospheric dynamics in S2S forecasting. We propose AnisoCast, which preserves complete latitude-ring context through anisotropic embedding and models zonal and meridional interactions with distinct operators. Their weight-shared symmetric composition couples global spectral mixing with local cross-latitude exchange while avoiding quadratic attention. Experiments on the ERA5 dataset show that AnisoCast achieves state-of-the-art forecasting accuracy and improved computational efficiency compared with NWP systems and data-driven models. We discuss its limitations in Appendix~\ref{appendix:limitations}. Future work will explore generative frameworks to quantify forecast uncertainty at chaotic S2S lead times.

\section*{AI Use Statement}
The authors used AI-assisted tools for grammar checking and language refinement.

\bibliographystyle{plainnat}
\bibliography{references}

\clearpage
\appendix

\begin{center}
    {\LARGE\bfseries Appendix (Supplementary Material)}
\end{center}
\vspace{3mm}
\startcontents[sections]
\noindent{\Large\scshape Table of Contents}\par
\vspace{2mm}
\hrule height .8pt
\vspace{2mm}
\printcontents[sections]{l}{1}{\setcounter{tocdepth}{2}}
\vspace{1mm}
\hrule height .8pt

\clearpage

\section{Related Works}
\label{sec:related_works}
\vspace{-2mm}
\subsection{Data-driven Weather Forecasting}
\vspace{-2mm}

Spatio-temporal data, together with time series data, are prevalent across various domains of daily life, among which weather data are closely related to human activities~\citep{UnravelingSpatio-TemporalFoundationModels,wang2024tsfool,SDformer}. Recently, deep learning has emerged as a computationally efficient alternative to traditional Numerical Weather Prediction (NWP) systems \citep{vitart2017subseasonal,white2017potential}. While early works focused on regional tasks~\citep{shi2015convolutional,harbola2019one}, the release of ERA5~\citep{hersbach2020era5} has shifted focus to global forecasting. Foundational benchmarks~\citep{rasp2020weatherbench} paved the way for models that rival operational NWP systems (e.g., IFS) in the medium range. Prominent architectures include Adaptive Fourier Neural Operators (FourCastNet~\citep{pathak2022fourcastnet}), Graph Neural Networks (GraphCast~\citep{lam2023learning}, AIFS~\citep{lang2024aifs}), and Transformers (PanguWeather~\citep{bi2023accurate}, FengWu~\citep{chen2023fengwu}, FuXi~\citep{chen2023fuxi}). Recent research also explores probabilistic forecasting via diffusion models~\citep{gencast}. However, the performance of these medium-range models often degrades at subseasonal scales due to the lack of physical constraints and error accumulation.

\subsection{Data-driven S2S Forecasting}
\vspace{-2mm}
Subseasonal-to-Seasonal (S2S) forecasting (2--6 weeks) remains a grand challenge due to chaotic atmospheric dynamics. While traditional dynamical models (e.g., ECMWF~\citep{molteni1996ecmwf}) face computational bottlenecks, recent global data-driven models have shown promise. ClimaX~\citep{nguyen2023climax} adapts a vision foundation model for climate tasks via fine-tuning. FuXi-S2S~\citep{FuXi-S2S} employs a cascade architecture to forecast daily averaged statistics. CirT~\citep{Cirt} introduces circular patching strategies to better handle spherical boundary conditions. EIMP~\citep{EIMP} leverages graph neural networks with rotation-equivariant message passing to model vector fields. TelePiT~\citep{telepit} combines spherical harmonics with physics-informed neural ODEs to capture multi-scale teleconnections, while TianQuan-S2S~\citep{TianQuans2s} further enhances predictability by incorporating climatology priors and noise injection. Although these approaches introduce different geometric, spectral, or physical inductive biases, they do not explicitly formulate zonal- and meridional-oriented latent transformations as a symmetric operator composition. In contrast, AnisoCast organizes these complementary transformations through a weight-shared $\mathcal{Z}\circ\mathcal{M}\circ\mathcal{Z}$ structure, providing an efficient direction-specific backbone for global S2S forecasting.

\section{Isotropic and Anisotropic Modeling Paradigms}
\label{app:paradigm}

\begin{figure*}[!ht]
    \centering
    \includegraphics[width=\textwidth, trim = 0 3 0 0, clip]{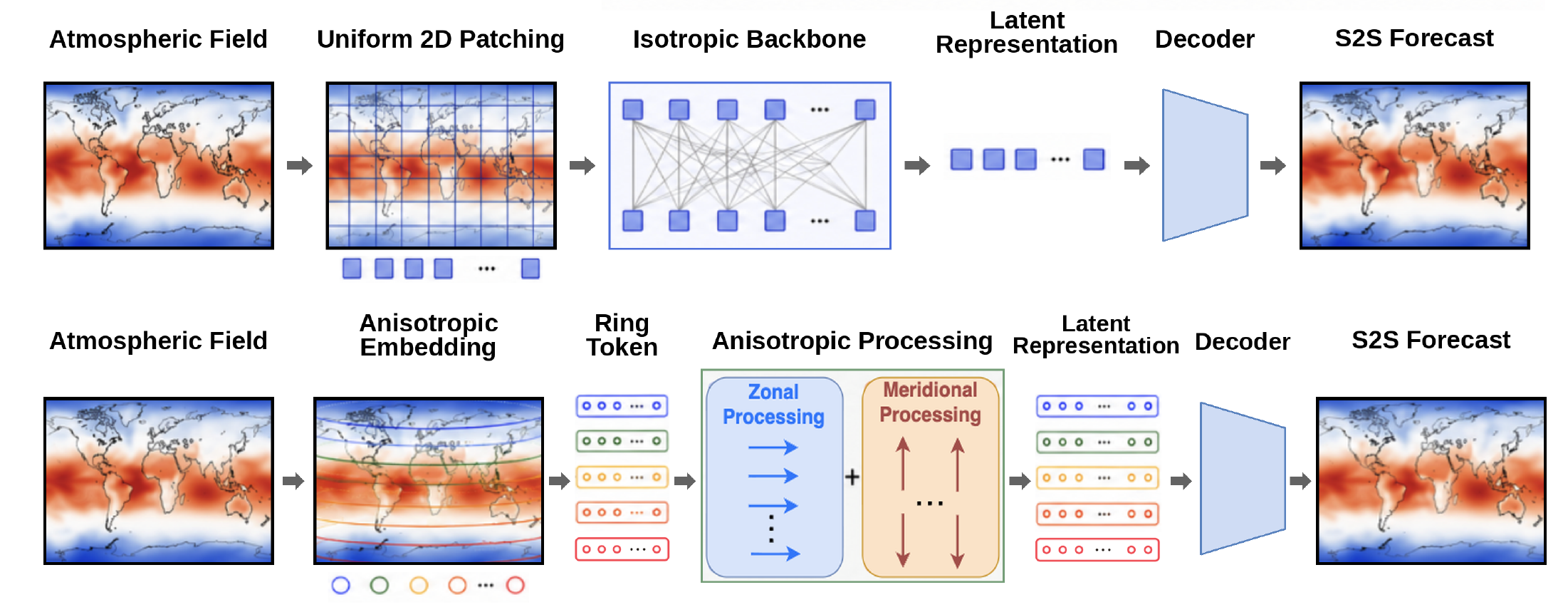}
    \caption{
    Paradigm-level comparison between isotropic modeling and the proposed
    anisotropic modeling for S2S forecasting.
    }
    \label{fig:paradigm_comparison}
\end{figure*}

Figure~\ref{fig:paradigm_comparison} holds the atmospheric input and S2S
forecast target fixed and illustrates the difference in spatial modeling
paradigms. The upper workflow summarizes the conventional isotropic modeling paradigm, in which latitude--longitude fields are processed
without explicitly decomposing zonal- and meridional-oriented latent
transformations. It is a schematic abstraction and does not imply that all
existing methods employ identical uniform 2D patches. Recent geometry-aware methods introduce circular patching, spherical message
passing, or spherical-harmonic representations~\citep{Cirt,EIMP,telepit}.
These designs account for important properties of global geometry, but they
do not explicitly organize zonal- and meridional-oriented latent
transformations. AnisoCast changes this spatial inductive bias rather than the forecasting
objective. Its Anisotropic Embedding maps each complete latitude ring to one
token, retaining the full longitudinal span within a single representation.
The subsequent processing assigns distinct roles to zonal-oriented
interactions within the learned ring representation and meridional-oriented
interactions across the ordered latitude sequence. This comparison concerns
architectural inductive biases and does not imply a unique decomposition of
the atmospheric governing equations.

\section{Channel-Spectral Analysis of LayerNorm and RMSNorm}
\label{appendix:proof_norm}
\label{appendix:proof_layernorm}
\label{appendix:proof_rmsnorm}

We analyze how normalization affects the channel spectrum of
a latent ring token $\bm{z}\in\mathbb{R}^{C}$.
Let $\mu=C^{-1}\sum_{n=0}^{C-1}z_n$ and
$\widehat{\bm{z}}=\mathcal{F}(\bm{z})$ denote its channel mean
and unnormalized channel-wise DFT, respectively.
Ignoring affine parameters, the normalization cores are
$\mathrm{LN}_0(\bm{z})=(\bm{z}-\mu\bm{1})/s_{\mathrm{LN}}$
and $\mathrm{RMS}_0(\bm{z})=\bm{z}/s_{\mathrm{RMS}}$,
where
$s_{\mathrm{LN}}=\sqrt{C^{-1}\sum_{n=0}^{C-1}(z_n-\mu)^2}$
and
$s_{\mathrm{RMS}}=\sqrt{C^{-1}\sum_{n=0}^{C-1}z_n^2}$
are assumed to be nonzero.
Since $\widehat{z}_0=C\mu$, subtracting the channel mean gives
$\mathcal{F}(\mathrm{LN}_0(\bm{z}))_0=0$,
while
$\mathcal{F}(\mathrm{LN}_0(\bm{z}))_k
=\widehat{z}_k/s_{\mathrm{LN}}$ for $k\neq0$.
LayerNorm therefore removes the channel-mean component,
which may contain information useful for forecasting.
In contrast, RMSNorm satisfies
$\mathcal{F}(\mathrm{RMS}_0(\bm{z}))_k
=\widehat{z}_k/s_{\mathrm{RMS}}$
for all $k=0,\ldots,C-1$, retaining the entire channel spectrum,
including its zero-frequency component, up to a common
positive scaling factor.
Both cores preserve the phases and relative magnitudes of
nonzero-frequency coefficients whenever these coefficients
are nonzero; their key difference is the removal or retention
of the channel-mean component.
These properties apply to the normalization cores, as learned
channel-wise gains and biases can further alter the spectrum.
They concern latent-channel coefficients and do not establish
preservation of all atmospheric information or physical wave phases.
The forecasting effect of the normalization choice is evaluated
in Section~\ref{subsec:ablation}.

\section{Theoretical Model Complexity}
\label{app:model_complexity}

We provide the detailed complexity derivation of AnisoCast. Let $H$ and $W$ denote the latitude and longitude resolutions, $K$ the number of meteorological variables, $C$ the latent dimension, $L$ the number of Aniso Blocks, and $k$ the meridional convolution kernel size. Standard Transformer-based models typically tokenize the latitude--longitude grid into $T\propto HW$ tokens. Therefore, global self-attention has complexity $\mathcal{O}(T^2C)=\mathcal{O}(H^2W^2C)$, which scales quadratically with the two-dimensional token count. In contrast, AnisoCast first applies Anisotropic Embedding to compress each latitudinal ring into one latent token. This embedding reads the original grid $\bm{X}_{t_0}\in\mathbb{R}^{K\times H\times W}$ and maps the $K$ input variables over $W$ longitudes into $C$ latent channels for each latitude, giving complexity $\mathcal{O}(HKWC)$. After this step, the latent representation has shape $\mathbb{R}^{H\times C}$, so the Aniso Block operates on only $H$ latitudinal ring tokens rather than $HW$ two-dimensional tokens. For each Aniso Block, the main costs come from channel-wise projections and FFN layers with complexity $\mathcal{O}(HC^2)$, latent spectral mixing in the Zonal Operator with complexity $\mathcal{O}(HC\log C)$, and meridional convolution with complexity $\mathcal{O}(HCk)$. Therefore, one Aniso Block costs $\mathcal{O}(H(C^2+C\log C+Ck))$, and $L$ stacked blocks cost $\mathcal{O}(LH(C^2+C\log C+Ck))$. Importantly, this backbone complexity is independent of $W$ because the longitude dimension has already been compressed into latent channels by the anisotropic embedding. Finally, the Decoder maps the latent representation back to the physical grid. Since each latitude token is projected from $C$ latent channels to longitude-resolved outputs over $K$ variables and $W$ longitudes, the decoder costs $\mathcal{O}(HCKW)$, with the constant factor for the two bi-weekly targets omitted in big-$\mathcal{O}$ notation. Combining the embedding, backbone, and decoder terms, the full complexity of AnisoCast is $\mathcal{O}(HKWC)+\mathcal{O}(LH(C^2+C\log C+Ck))+\mathcal{O}(HCKW)$. Both AnisoCast and CirT~\citep{Cirt} compress the longitude
dimension into ring-token features. AnisoCast further avoids the $\mathcal{O}(LH^2C)$ attention
term through its spectral and convolutional backbone.

\section{Experimental Details}

\subsection{Dataset}
\label{app:datasets}

We evaluate AnisoCast on the ERA5 reanalysis dataset~\citep{hersbach2020era5}, which provides a comprehensive record of global atmospheric conditions.
To construct the training and evaluation data for our S2S forecasting task, we downsample the original $0.25^{\circ}$ ERA5 data to a coarser resolution of $1.5^{\circ}\times1.5^{\circ}$ by sub-sampling grid points that align with the target grid.
Specifically, we extract coordinates $(\varphi,\lambda)$ on the latitude--longitude domain
$\Omega=[-90^{\circ},90^{\circ}]\times[-180^{\circ},180^{\circ})$,
where $\varphi\in\{-90^{\circ},-88.5^{\circ},\ldots,90^{\circ}\}$ and
$\lambda\in\{-180^{\circ},-178.5^{\circ},\ldots,178.5^{\circ}\}$,
resulting in a spatial resolution of $121\times240$.
This regridding strategy ensures that our input data structure is consistent with the outputs of baseline models such as PanguWeather~\citep{bi2023accurate} and FourCastNetV2~\citep{pathak2022fourcastnet} when evaluated at this resolution.
We select a set of 63 variables critical for capturing atmospheric dynamics, covering the time period from 1979 to 2018 with a temporal resolution of 1 day (daily averages).
A detailed summary of the dataset configuration is provided in Table~\ref{tab:dataset_config}.

\begin{table*}[!ht]
\caption{Detailed configuration of the ERA5 dataset used in this work.}
\label{tab:dataset_config}
\vspace{-2mm}
\centering
\footnotesize
\setlength{\tabcolsep}{1.2pt}
\renewcommand{\arraystretch}{1.2}
\scalebox{0.96}{
\begin{tabular}{c c c c c c c}
\toprule
\textbf{Category} & \textbf{Variable Name} & \textbf{Symbol} & \textbf{Levels} & \textbf{Resolution} & \textbf{Lat-Lon Range} & \textbf{Time} \\
\midrule
\multirow{6}{*}{\textbf{Pressure Level}}
& Geopotential & $z$ & 10 & $1.5^{\circ}$ & $-90^{\circ}\text{S}\;180^{\circ}\text{W} \sim 90^{\circ}\text{N}\;180^{\circ}\text{E}$ & $1979 \sim 2018$ \\
& Specific Humidity & $q$ & 10 & $1.5^{\circ}$ & $-90^{\circ}\text{S}\;180^{\circ}\text{W} \sim 90^{\circ}\text{N}\;180^{\circ}\text{E}$ & $1979 \sim 2018$ \\
& Temperature & $t$ & 10 & $1.5^{\circ}$ & $-90^{\circ}\text{S}\;180^{\circ}\text{W} \sim 90^{\circ}\text{N}\;180^{\circ}\text{E}$ & $1979 \sim 2018$ \\
& U Component of Wind & $u$ & 10 & $1.5^{\circ}$ & $-90^{\circ}\text{S}\;180^{\circ}\text{W} \sim 90^{\circ}\text{N}\;180^{\circ}\text{E}$ & $1979 \sim 2018$ \\
& V Component of Wind & $v$ & 10 & $1.5^{\circ}$ & $-90^{\circ}\text{S}\;180^{\circ}\text{W} \sim 90^{\circ}\text{N}\;180^{\circ}\text{E}$ & $1979 \sim 2018$ \\
& Vertical Velocity & $w$ & 10 & $1.5^{\circ}$ & $-90^{\circ}\text{S}\;180^{\circ}\text{W} \sim 90^{\circ}\text{N}\;180^{\circ}\text{E}$ & $1979 \sim 2018$ \\
\cmidrule{1-7}
\multirow{3}{*}{\textbf{Single Level}}
& 2 Metre Temperature & $t2m$ & 1 & $1.5^{\circ}$ & $-90^{\circ}\text{S}\;180^{\circ}\text{W} \sim 90^{\circ}\text{N}\;180^{\circ}\text{E}$ & $1979 \sim 2018$ \\
& 10 Metre U Wind & $u10$ & 1 & $1.5^{\circ}$ & $-90^{\circ}\text{S}\;180^{\circ}\text{W} \sim 90^{\circ}\text{N}\;180^{\circ}\text{E}$ & $1979 \sim 2018$ \\
& 10 Metre V Wind & $v10$ & 1 & $1.5^{\circ}$ & $-90^{\circ}\text{S}\;180^{\circ}\text{W} \sim 90^{\circ}\text{N}\;180^{\circ}\text{E}$ & $1979 \sim 2018$ \\
\bottomrule
\end{tabular}
}
\end{table*}

\subsection{Baselines}
\label{app:baselines}

To rigorously evaluate the performance of AnisoCast, we benchmark it against a comprehensive set of baseline models:

\noindent\textbf{Operational Forecasting Systems.} These established numerical weather prediction (NWP) systems represent the current gold standard in operational meteorological forecasting.

\begin{itemize}[leftmargin=*]
    \item \textbf{CMA:} The China Meteorological Administration implements the Beijing Climate Center's fully-coupled BCC-CSM2-HR model~\citep{wu2019beijing}, providing operational forecasts with a 60-day lead time.
    \item \textbf{NCEP:} The National Centers for Environmental Prediction deploys the Climate Forecast System version 2 (CFSv2)~\citep{saha2014ncep} to generate daily control forecasts extending to a 45-day prediction window.
    \item \textbf{ECMWF:} The European Centre for Medium-Range Weather Forecasts employs the Integrated Forecasting System (IFS)~\citep{molteni1996ecmwf}.
\end{itemize}

For NWP baselines, we follow the variable availability of the standardized S2S evaluation archives used in prior work~\citep{telepit}. CMA, NCEP, and ECMWF do not provide $u10$/$v10$ forecasts in this setting; this is a data-availability constraint of the archived products.

\noindent\textbf{Data-Driven Models.} We compare against both pre-trained foundation models and models trained from scratch.
\begin{itemize}[leftmargin=*]
    \item \textbf{PanguWeather~\citep{bi2023accurate}:} A 3D Earth-specific Transformer that tokenizes pressure-level and surface data separately, trained using a hierarchical temporal aggregation strategy. We use the API \href{https://github.com/ecmwf-lab/ai-models}{https://github.com/ecmwf-lab/ai-models} to perform inference.
    \item \textbf{FourCastNetV2 (FCN2)~\citep{pathak2022fourcastnet}:} An iterative model based on Vision Transformers that incorporates Adaptive Fourier Neural Operators (AFNO) for efficient spatial token mixing. We use the API \href{https://github.com/ecmwf-lab/ai-models}{https://github.com/ecmwf-lab/ai-models} to perform inference.
    \item \textbf{ClimaX~\citep{nguyen2023climax}:} A foundation model for weather and climate using a Vision Transformer backbone with variable-specific tokenization and aggregation.
    \item \textbf{ClimODE~\citep{verma2024climode}:} A physics-informed approach that models atmospheric dynamics using Neural Ordinary Differential Equations (ODEs) to capture continuous-time evolution. We download the official code from \href{https://github.com/Aalto-QuML/ClimODE}{https://github.com/Aalto-QuML/ClimODE}.
    \item \textbf{CirT~\citep{Cirt}:} A Transformer-based model designed for S2S forecasting that introduces circular patching and frequency-domain mixing to address spherical boundary conditions. We download the official code from \href{https://github.com/compasszzn/CirT}{https://github.com/compasszzn/CirT}.
    \item \textbf{EIMP~\citep{EIMP}:} A graph-based model utilizing equivariant and invariant message passing to strictly enforce rotational symmetries on spherical grids. We download the official code from \href{https://zenodo.org/records/15516468}{https://zenodo.org/records/15516468}.
    \item \textbf{TelePiT~\citep{telepit}:} A physics-informed Transformer that integrates spherical harmonic embeddings with neural ODEs to explicitly capture multi-scale teleconnections. We download the official code from \href{https://github.com/usail-hkust/TelePiT}{https://github.com/usail-hkust/TelePiT}.
    \item \textbf{TianQuan-S2S~\citep{TianQuans2s}:}
    We use its official repository\footnote{\url{https://github.com/zhangminglang42/TianQuan}}
    and aggregate daily-mean forecasts over the target days to obtain
    weeks~3--4 and weeks~5--6 means.

    \item \textbf{OmniCast~\citep{nguyen2025omnicast}:}
    We use its official repository
    \footnote{\url{https://github.com/tung-nd/omnicast}},
    spatially align its outputs to the benchmark grid, and compute
    the same latitude-weighted metrics on the shared verification targets.
\end{itemize}

Since PanguWeather~\citep{bi2023accurate} and FourCastNetV2~\citep{pathak2022fourcastnet} are pretrained medium-range autoregressive models, we roll them out to the required target days and average the generated fields into the same bi-weekly verification windows. Therefore, they serve as autoregressive reference baselines under our standardized S2S evaluation protocol, rather than as models trained with the same direct bi-weekly objective as AnisoCast. For controlled comparison, all trainable learned baselines in the main tables are evaluated under our unified preprocessing, inference, and metric-computation protocol. For data-driven models, we do not mix originally reported scores from prior papers into the main benchmark because public repositories differ in preprocessing details, checkpoint availability, or inference pipelines; such reported numbers are therefore treated as non-directly comparable to the unified results.

Moreover, our reproduced CirT scores differ from Table~1 of the original
paper~\citep{Cirt}.
The released repository
\footnote{\url{https://github.com/compasszzn/CirT}}
leaves key evaluation choices unspecified:
\texttt{test\_step} reports a single normalized, tensor-level RMSE
without de-normalization or variable-wise RMSE/ACC aggregation,
and the available ACC function is not called.
The default configuration also reverses the reported validation/test
years, while the climatology script does not restrict input years.
Together with unavailable published checkpoints and frozen
preprocessing artifacts, these omissions prevent exact reconstruction
of the published evaluation. We therefore evaluate CirT using the same pipeline as the other
baselines: the fixed 1979--2016/2017/2018 train/validation/test split,
training-only normalization and climatology, and common physical-unit
RMSE and ACC aggregation.
Our released code
\footnote{\url{https://github.com/zhouziyu02/AnisoCast}} records these settings and verifies frozen inputs and result tables with SHA-256 manifests.
This makes the comparison in Table~\ref{tab:mainresults} consistent across models while clarifying why its CirT scores may differ from the originally reported values.

Following CirT~\citep{Cirt}, we compute latitude-weighted ACC
on the 14-day mean fields for weeks~3--4 and weeks~5--6.
For each evaluated variable, anomalies are obtained by
subtracting the same observational climatology from
predictions and targets.
In our evaluation, this climatology is estimated from
the 1979--2016 training data and is shared across all models,
ensuring a consistent reference for the comparison. Their absolute values should be interpreted within this setting.
Comparisons with WeatherBench~\citep{rasp2020weatherbench}
or ECMWF forecasting evaluations~\citep{molteni1996ecmwf}
require accounting for differences in forecast targets,
temporal aggregation, spatial domains, seasons, variables,
and climatology definitions.

\subsection{Implementation Details}
\label{app:implementation}

The AnisoCast model is implemented using PyTorch 2.3.1 and trained on 8 NVIDIA GPUs via Distributed Data Parallel (DDP) with mixed precision (FP16), while FFT operations within the Zonal Operator are forced to FP32 to ensure numerical stability. We utilize an anisotropic embedding kernel size of $(1, 240)$ and a meridional spatial kernel size of $k=7$. The model is trained for 20 epochs with a global batch size of 256 using the AdamW optimizer (weight decay $10^{-5}$, $\epsilon=10^{-8}$) and gradient clipping at 1.0. The learning rate follows a cosine annealing schedule, decaying from a peak of $10^{-3}$ to $10^{-4}$. To determine the optimal architecture, we conducted a grid search on the validation set over the following hyperparameter spaces: embedding dimension $C \in \{128, 192, 256, 384, 512\}$, Aniso Blocks $L \in \{4, 5, 6, 7, 8, 9\}$, decoder depth $\in \{1, 2, 3\}$, learning rate $\in \{5\times10^{-4}, 1\times10^{-3}, 5\times10^{-3}\}$, MLP ratio $\in \{2.0, 4.0, 6.0\}$, and dropout rates $\in \{0.0, 0.1, 0.2\}$. The final selected configuration is $C=256$, $L=7$, decoder depth 1, learning rate $10^{-3}$, MLP ratio 4.0, and dropout 0.1, resulting in a parameter count of approximately 16.58M.

\subsection{Ablation Settings}
\label{app:ablation_settings}

The following
variants  are studied in the
ablation studies
in Section~\ref{subsec:ablation}:
\begin{enumerate}[leftmargin=*]
\item \textbf{Variant without Anisotropic Embedding (w/o AE).}
We replace Anisotropic Embedding with an order-agnostic ring embedding constructed specifically as a shape-matched control.
Formally, the replacement embedding is defined as $$\widetilde{\bm{E}}^{(0)}=\mathrm{RMSN}\left(\mathrm{Flat}\left(\mathrm{Linear}\left(\mathrm{AvgPool}_{1\times W}(\bm{X}_{t_0})\right)\right)\right)\in\mathbb{R}^{H\times C}.$$ $\mathrm{AvgPool}_{1\times W}$ averages each latitude ring over longitude, and $\mathrm{Linear}$ maps the resulting $K$-dimensional vector at each latitude to $C$ channels. This control evaluates the effect of removing
longitude-dependent input information, rather than isolating the contribution of the embedding geometry.

\item \textbf{Variant without Zonal Operator (w/o ZO).} We remove the spectral mechanism and replace the $\mathcal{Z} \circ \mathcal{M} \circ \mathcal{Z}$ symmetric structure with stacked Meridional Operators ($\mathcal{M} \circ \mathcal{M}$).
\item \textbf{Variant without Meridional Operator (w/o MO).} We replace the spatial operator with a double Zonal structure ($\mathcal{Z} \circ \mathcal{Z}$).
\item
\textbf{Variant with unshared-weight Zonal Operator (w/ unshared ZO).} We define two independent Zonal Operators, $\mathcal{Z}_1$ and $\mathcal{Z}_2$, resulting in an asymmetric structure $\mathcal{Z}_1 \circ \mathcal{M} \circ \mathcal{Z}_2$.
\item \textbf{Variant with LayerNorm (w/ LN).} We replace all RMSNorm layers with standard Layer Normalization to evaluate the impact of mean-centering on long-term stability.
\item \textbf{Variant with alternative symmetric ordering.}
We replace $\mathcal{Z}\circ\mathcal{M}\circ\mathcal{Z}$ with
$\mathcal{M}\circ\mathcal{Z}\circ\mathcal{M}$, where the two outer Meridional Operators share weights, while keeping all other settings unchanged.
\end{enumerate}

\section{Additional Results on Robustness and Generalization}
\label{app:robustness}

\subsection{Multi-seed Robustness}
\label{app:seed}

To evaluate sensitivity to random initialization, we train AnisoCast
with five random seeds and report the mean and standard deviation
of RMSE and ACC in Table~\ref{tab:multi_seed}.

\begin{table}[!ht]
\centering
\caption{Multi-seed robustness of AnisoCast over five random seeds. We report mean$\pm$std for RMSE and ACC on weeks 3--4 and weeks 5--6. Lower RMSE and higher ACC indicate better performance.}
\vspace{-2mm}
\label{tab:multi_seed}
\resizebox{\linewidth}{!}{
\begin{tabular}{c|cccc|cccc}
\toprule
\textbf{Period}
& \textbf{$z500$ RMSE}
& \textbf{$z850$ RMSE}
& \textbf{$u10$ RMSE}
& \textbf{$v10$ RMSE}
& \textbf{$z500$ ACC}
& \textbf{$z850$ ACC}
& \textbf{$u10$ ACC}
& \textbf{$v10$ ACC} \\
\midrule
w3--4
& $51.993{\pm}0.233$
& $33.755{\pm}0.412$
& $1.995{\pm}0.016$
& $1.719{\pm}0.143$
& $0.981{\pm}0.002$
& $0.958{\pm}0.002$
& $0.871{\pm}0.001$
& $0.783{\pm}0.002$ \\
w5--6
& $52.518{\pm}0.275$
& $33.400{\pm}0.382$
& $1.995{\pm}0.090$
& $1.671{\pm}0.114$
& $0.981{\pm}0.002$
& $0.958{\pm}0.002$
& $0.872{\pm}0.002$
& $0.786{\pm}0.002$ \\
\bottomrule
\end{tabular}
}
\end{table}

Across both forecast windows, AnisoCast's mean RMSE remains below
the reported baseline values on all four evaluated variables,
while ACC standard deviations are at most $0.002$.
These small ACC variations indicate limited sensitivity to random
initialization, supporting robustness across training runs.
The larger relative variation in $v10$ RMSE nevertheless indicates
greater seed sensitivity for near-surface wind errors.

\subsection{Extended Evaluation Period for Interannual Robustness}
\label{app:extended-period}

To demonstrate that AnisoCast's performance improvements are robust to interannual climate variability and not an artifact of the specific 2018 test period, we extend our evaluation to encompass both 2018 and 2019. S2S atmospheric dynamics often exhibit significant year-to-year shifts driven by broader Earth system signals (e.g., ENSO phases~\citep{wang2016nino}). Expanding the test horizon allows us to examine whether the model has successfully learned generalized physical laws rather than overfitting to a single year's weather regimes.

\begin{table*}[h]
\centering
\footnotesize
\setlength{\tabcolsep}{3pt}
\renewcommand{\arraystretch}{1.05}
\caption{Extended-period evaluation on the 2018--2019 held-out years.
Scores are averaged over the weeks 3--4 and weeks 5--6 forecast windows.
The best result across all models is in \textbf{bold}, and the best
baseline result is \underline{underlined}.
Lower RMSE and higher ACC indicate better performance.
ACC is dimensionless. The variable groups follow
Table~\ref{tab:mainresults}.}
\label{tab:extended-period}
\begin{threeparttable}
\resizebox{\linewidth}{!}{
\begin{tabular}{c cc | *{4}{c} | *{3}{c} | *{4}{c} | *{3}{c}}
\toprule
  & \multicolumn{2}{c|}{\multirow{3}{*}{\textbf{Model}}}
  & \multicolumn{7}{c|}{\textbf{RMSE} ($\downarrow$)}
  & \multicolumn{7}{c}{\textbf{ACC} ($\uparrow$)} \\
\cmidrule(lr){4-10}\cmidrule(lr){11-17}
  & &
  & \multicolumn{4}{c|}{\textbf{Pressure Level}}
  & \multicolumn{3}{c|}{\textbf{Single Level}}
  & \multicolumn{4}{c|}{\textbf{Pressure Level}}
  & \multicolumn{3}{c}{\textbf{Single Level}} \\
\cmidrule(lr){4-7}\cmidrule(lr){8-10}
\cmidrule(lr){11-14}\cmidrule(lr){15-17}
  & &
   & $z500$ & $z850$ & $t500$ & $t850$
  & $t2m$ & $u10$ & $v10$
  & $z500$ & $z850$ & $t500$ & $t850$
  & $t2m$ & $u10$ & $v10$ \\
\midrule
\multirow{8}{*}{\rotatebox{90}{\textbf{2018--2019}}}
  & \multicolumn{2}{c|}{ClimaX}
  & 67.287 & 40.132 & 2.293 & 2.692 & 5.298 & 3.268
  & 2.665
  & 0.960 & 0.948 & 0.970 & 0.968 & 0.955 & 0.806 & 0.641 \\
  & \multicolumn{2}{c|}{ClimODE}
  & 70.326 & 42.260 & 2.665 & 3.166 & 5.122 & 4.076
  & 2.892
  & 0.957 & 0.947 & 0.968 & 0.965 & 0.958 & 0.828 & 0.735 \\
  & \multicolumn{2}{c|}{CirT}
  & 65.211 & 38.544 & 2.302 & 2.987 & 4.696 & 2.451
  & \underline{1.989}
  & 0.968 & 0.948 & 0.974 & 0.982 & \underline{0.973}
  & 0.831 & 0.709 \\
  & \multicolumn{2}{c|}{EIMP}
  & 66.724 & 39.138 & 2.465 & \underline{2.489} & 4.592
  & 2.512 & 2.332
  & 0.974 & 0.950 & 0.975 & \textbf{\underline{0.983}}
  & 0.972 & 0.798 & 0.732 \\
  & \multicolumn{2}{c|}{TelePiT}
  & \underline{64.016} & \underline{38.246} & 2.320 & 2.677
  & \underline{4.311} & \underline{2.392} & 2.264
  & 0.972 & 0.949 & 0.976 & 0.978 & 0.962
  & \underline{0.853} & 0.739 \\
  & \multicolumn{2}{c|}{TianQuan-S2S}
  & 65.981 & 41.022 & 2.412 & 2.774 & 4.599 & 2.522
  & 2.032
  & 0.975 & \underline{0.953} & 0.975 & 0.973 & 0.963
  & 0.837 & \underline{0.764} \\
  & \multicolumn{2}{c|}{OmniCast}
  & 66.403 & 39.293 & \underline{2.121} & 2.629 & 4.482
  & 2.496 & 2.091
  & \underline{0.977} & 0.952 & \underline{0.979} & 0.973
  & 0.963 & 0.839 & 0.763 \\
  & \multicolumn{2}{c|}{\textbf{AnisoCast}}
  & \textbf{62.381} & \textbf{36.489} & \textbf{1.936}
  & \textbf{2.442} & \textbf{4.298} & \textbf{2.233}
  & \textbf{1.899}
  & \textbf{0.978} & \textbf{0.954} & \textbf{0.982}
  & 0.982 & \textbf{0.977} & \textbf{0.869}
  & \textbf{0.783} \\
\bottomrule
\end{tabular}}
\end{threeparttable}
\end{table*}

As shown in Table~\ref{tab:extended-period}, AnisoCast consistently achieves the lowest RMSE across all reported variables under this extended evaluation period. Quantitatively, it reduces the RMSE relative to the strongest respective baseline by approximately 2.6\% for upper-air circulation ($z500$), 8.7\% for mid-troposphere temperature ($t500$), and 6.6\% and 4.5\% for near-surface winds ($u10$ and $v10$). AnisoCast also achieves the highest ACC on six of the seven
reported variables, while EIMP performs slightly better on $t850$. These results carry an important implication: the geometric and physical inductive biases introduced by AnisoCast's anisotropic modeling yield robust forecasting skill that persists through different climatic backgrounds. It verifies that our performance gains are not transient phenomena tied to a specific temporal window. We maintain 2018 as the primary benchmark in the main text purely to ensure strictly controlled comparisons with existing direct S2S literature.

\subsection{Seasonal Robustness Analysis}
\label{app:monthlyvisual}
S2S forecasting skill often varies significantly by season due to changing atmospheric dynamics. To evaluate robustness, we break down the testing set performance by month. Figure~\ref{fig:monthlyz500} to Figure~\ref{fig:monthlyt850} illustrate the monthly RMSE for key variables ($z500, z850, t500, t850$). AnisoCast remains competitive across months, with particularly low errors for $t500$ and $t850$, although rankings vary by variable and month. Whether in Northern Hemisphere winter (where baroclinic instability is high) or summer, AnisoCast maintains the lowest error margins. This suggests that the model has successfully learned generalized representations of seasonal cycle dynamics rather than overfitting to specific regimes.

\begin{figure*}[!ht]
    \centering
    \includegraphics[width=\linewidth]{Figure/z500_monthly_final.pdf}
    \vspace{-6mm}
    \caption{Monthly RMSE breakdown for $z500$.}
    \label{fig:monthlyz500}
\end{figure*}

\begin{figure*}[!ht]
    \centering
    \includegraphics[width=\linewidth]{Figure/z850_monthly_final.pdf}
    \vspace{-6mm}
    \caption{Monthly RMSE breakdown for $z850$.}
    \label{fig:monthlyz850}
\end{figure*}

\begin{figure*}[!ht]
    \centering
    \includegraphics[width=\linewidth]{Figure/t500_monthly_final.pdf}
    \vspace{-6mm}
    \caption{Monthly RMSE breakdown for $t500$.}
    \label{fig:monthlyt500}
\end{figure*}

\begin{figure*}[!ht]
    \centering
    \includegraphics[width=\linewidth]{Figure/t850_monthly_final.pdf}
    \vspace{-6mm}
    \caption{Monthly RMSE breakdown for $t850$.}
    \label{fig:monthlyt850}
\end{figure*}

\subsection{Evaluation under Additional Criteria}
\label{app:additional_criteria}

In addition to RMSE and spatial ACC reported in the main text, we further evaluate AnisoCast using Temporal Correlation Coefficient (TCC) and extreme-event metrics. Following FuXi-S2S~\citep{FuXi-S2S}, TCC measures the temporal phase consistency of anomaly forecasts. For each variable and lead window, let \(a_{n,i}\) and \(\hat{a}_{n,i}\) denote the observed and predicted anomalies at evaluation sample \(n\) and grid point \(i\), and let \(\alpha_i\) be the latitude weight. We compute
\[
\rho_i =
\frac{
\sum_n (\hat{a}_{n,i}-\bar{\hat{a}}_i)(a_{n,i}-\bar{a}_i)
}{
\sqrt{\sum_n (\hat{a}_{n,i}-\bar{\hat{a}}_i)^2}
\sqrt{\sum_n (a_{n,i}-\bar{a}_i)^2}
},
\qquad
\mathrm{TCC} =
\frac{\sum_i \alpha_i \rho_i}{\sum_i \alpha_i}.
\]
We report TCC on representative atmospheric variables and total precipitation (\texttt{tp}),an additional key surface variable for climate forecasting ~\citep{chen2023fengwu,chen2023fuxi,FuXi-S2S,baguan}.

\begin{table}[!ht]
\centering
\caption{Evaluation under additional criteria. We report TCC for $z500$, $z850$, $t2m$, $v10$, and total precipitation (\texttt{tp}). CSI and SEDI are reported for extreme-event evaluation on $u10$, $v10$, and $t2m$. All metrics are dimensionless, and higher values indicate better performance.}
\label{tab:additional_criteria}
\scriptsize
\setlength{\tabcolsep}{2.3pt}
\resizebox{\textwidth}{!}{
\begin{tabular}{c|ccccc|ccccc|ccc|ccc}
\toprule
\multirow{2}{*}{\textbf{Model}}
& \multicolumn{5}{c|}{\textbf{TCC W3--4}}
& \multicolumn{5}{c|}{\textbf{TCC W5--6}}
& \multicolumn{3}{c|}{\textbf{CSI}}
& \multicolumn{3}{c}{\textbf{SEDI}} \\
\cmidrule(lr){2-6}
\cmidrule(lr){7-11}
\cmidrule(lr){12-14}
\cmidrule(lr){15-17}
& $z500$ & $z850$ & $t2m$ & $v10$ & \texttt{tp}
& $z500$ & $z850$ & $t2m$ & $v10$ & \texttt{tp}
& $u10$ & $v10$ & $t2m$
& $u10$ & $v10$ & $t2m$ \\
\midrule
CirT
& 0.420 & 0.419 & 0.365 & 0.301 & 0.261
& 0.351 & 0.350 & 0.370 & 0.303 & 0.169
& 0.09  & 0.12 & 0.22
& 0.29  & 0.33 & 0.40 \\
EIMP
& 0.412 & 0.411 & 0.368 & 0.292 & 0.251
& 0.334 & 0.331 & 0.360 & \textbf{0.313} & 0.177
& 0.10  & 0.12 & 0.22
& 0.31  & 0.34 & 0.39 \\
TelePiT
& 0.428 & 0.424 & 0.370 & 0.289 & 0.276
& 0.362 & 0.358 & \textbf{0.372} & 0.312 & 0.182
& 0.10  & 0.11 & 0.22
& 0.32  & 0.32 & 0.38 \\
AnisoCast
& \textbf{0.442} & \textbf{0.440} & \textbf{0.372} & \textbf{0.313} & \textbf{0.282}
& \textbf{0.373} & \textbf{0.370} & 0.368 & 0.309 & \textbf{0.196}
& \textbf{0.12}  & \textbf{0.13} & \textbf{0.24}
& \textbf{0.34}  & \textbf{0.35} & \textbf{0.42} \\
\bottomrule
\end{tabular}
}
\end{table}

For extreme-event evaluation, we report Critical Success Index (CSI)~\citep{Ocean-E2E} and Symmetric Extremal Dependence Index (SEDI)~\citep{gao2025oneforecast}. Extreme events are defined using the same grid-wise threshold \(k_i\) for all models. In our implementation, \(k_i\) is computed as the 90th percentile of the training-period climatology at each grid point and for each variable. We compute
\begin{equation*}
\mathrm{CSI}=\frac{\mathrm{TP}}{\mathrm{TP}+\mathrm{FP}+\mathrm{FN}}, \quad
\mathrm{SEDI}=
\frac{\log(F)-\log(H)-\log(1-F)+\log(1-H)}
{\log(F)+\log(H)+\log(1-F)+\log(1-H)},
\end{equation*}

where \(H=\frac{\mathrm{TP}}{\mathrm{TP}+\mathrm{FN}}\) is the hit rate and \(F=\frac{\mathrm{FP}}{\mathrm{FP}+\mathrm{TN}}\) is the false alarm rate. Higher TCC, CSI, and SEDI indicate better performance.

As shown in Table~\ref{tab:additional_criteria}, AnisoCast achieves
the best TCC on $z500$, $z850$, and total precipitation in both
forecast windows, while remaining competitive on near-surface variables.
It also achieves the best CSI and SEDI on $u10$, $v10$, and $t2m$.
These gains suggest that its learned representations capture
temporal variations and extreme atmospheric states across different
variables, supporting broader applicability beyond the primary
evaluation variables and metrics.

\subsection{Multi-resolution Evaluation}
\label{app:resolution}

To examine whether the performance gain of AnisoCast is sensitive to the spatial resolution, we further evaluate different models under three grid settings: $5.625^\circ$ ($32\times64$), $1.5^\circ$ ($121\times240$), and $0.25^\circ$ ($721\times1440$). The results are summarized in Table~\ref{tab:multi_resolution}.

\begin{table}[!ht]
\centering
\caption{Multi-resolution evaluation on Weeks 3--4. We compare different models under $5.625^\circ$, $1.5^\circ$, and $0.25^\circ$ spatial resolutions. The best results under each resolution are highlighted in bold. Lower RMSE and higher ACC indicate better performance.}
\label{tab:multi_resolution}
\scalebox{0.88}{
\begin{tabular}{c|c|cc|cc}
\toprule
\textbf{Model}
& \textbf{Resolution}
& \textbf{$z500$ RMSE}
& \textbf{$t2m$ RMSE}
& \textbf{$z500$ ACC}
& \textbf{$t2m$ ACC} \\
\midrule
CirT
& $5.625^\circ$
& 56.103
& 5.510
& 0.978
& 0.957 \\
EIMP
& $5.625^\circ$
& 66.326
& 4.783
& 0.980
& 0.974 \\
TelePiT
& $5.625^\circ$
& 59.780
& 4.121
& 0.976
& 0.973 \\
AnisoCast
& $5.625^\circ$
& \textbf{52.015}
& \textbf{3.320}
& \textbf{0.981}
& \textbf{0.983} \\
\midrule
CirT
& $1.5^\circ$
& 55.505
& 5.478
& 0.980
& 0.958 \\
EIMP
& $1.5^\circ$
& 66.498
& 4.785
& 0.981
& 0.974 \\
TelePiT
& $1.5^\circ$
& 59.772
& 4.002
& 0.977
& 0.974 \\
AnisoCast
& $1.5^\circ$
& \textbf{51.956}
& \textbf{3.255}
& \textbf{0.982}
& \textbf{0.982} \\
\midrule
CirT
& $0.25^\circ$
& 56.124
& 5.682
& 0.977
& 0.956 \\
EIMP
& $0.25^\circ$
& 67.589
& 4.851
& 0.980
& 0.972 \\
TelePiT
& $0.25^\circ$
& 58.761
& 3.988
& 0.977
& 0.975 \\
AnisoCast
& $0.25^\circ$
& \textbf{51.262}
& \textbf{3.188}
& \textbf{0.983}
& \textbf{0.982} \\
\bottomrule
\end{tabular}
}
\end{table}

Across all resolutions, AnisoCast consistently achieves the best performance on both $z500$ and $t2m$, indicating that its advantage does not rely on a particular grid configuration. We also observe that increasing the spatial resolution does not uniformly improve S2S forecasting skill, which is consistent with the long-range nature of the task: at subseasonal horizons, predictability is often governed more by large-scale circulation patterns than by fine-scale spatial details. These results support our use of the $1.5^\circ$ setting as a practical and effective trade-off between computational cost and forecasting accuracy. This analysis is based on ERA5 verification and should not be interpreted as independent polar validation against radiosonde or satellite observations.

\subsection{Evaluation Excluding High-latitude Regions}

\label{app:non_polar_eval}

\begin{figure}[!ht]
    \centering
    \includegraphics[width=\textwidth]{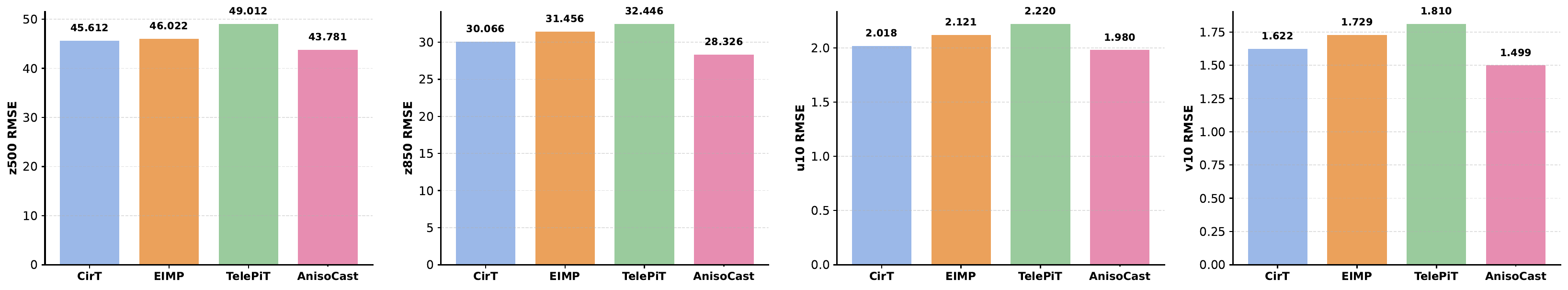}
    \caption{Evaluation over the non-polar latitude band $|\phi|\leq 60^\circ$, excluding high-latitude regions.}
    \label{fig:non_polar_eval}
\end{figure}

To verify that the overall improvement of AnisoCast is not mainly driven by polar-region performance, we additionally evaluate all models over the non-polar latitude band $|\phi|\leq 60^\circ$, excluding high-latitude regions above $60^\circ$. The results are visualized in Figure~\ref{fig:non_polar_eval}. Each bar reports the region-aggregated RMSE for one variable, averaged over the selected latitude band, all longitudes, and the evaluation samples. AnisoCast achieves the lowest RMSE across all four reported variables, including $z500$, $z850$, $u10$, and $v10$. Specifically, AnisoCast obtains $z500$ RMSE of $43.781$, $z850$ RMSE of $28.326$, $u10$ RMSE of $1.980$, and $v10$ RMSE of $1.499$, outperforming CirT, EIMP, and TelePiT. These results indicate that the advantage of AnisoCast does not rely solely on the high-latitude or polar regions, but remains evident under the standard S2S evaluation protocol after those regions are excluded.

\section{Additional Results on Architectural and Forecasting Choices}
\label{app:design_choices}
\subsection{Full Ablation Results}
\label{app:full_ablation}

\begin{table}[!ht]
\centering
\footnotesize
\setlength{\tabcolsep}{1.8pt}
\renewcommand{\arraystretch}{1.1}
\caption{Full ablation study results of AnisoCast on weeks 3--4 and weeks 5--6.}
\vspace{-2mm}
\label{tab:ablation}
\begin{threeparttable}
\newcommand{\best}[1]{\textbf{#1}}
\newcommand{\second}[1]{\underline{#1}}
\scalebox{0.61}{
\begin{tabular}{c cc | *{4}{c} | *{3}{c} | *{4}{c} | *{3}{c}}
\toprule
  & \multicolumn{2}{c|}{\multirow{3}{*}{\textbf{Model}}}
  & \multicolumn{7}{c|}{\textbf{RMSE} ($\downarrow$)}
  & \multicolumn{7}{c}{\textbf{ACC} ($\uparrow$)} \\
\cmidrule(lr){4-10}\cmidrule(lr){11-17}
  &  &
  & \multicolumn{4}{c|}{\textbf{Pressure Level}}
  & \multicolumn{3}{c|}{\textbf{Single Level}}
  & \multicolumn{4}{c|}{\textbf{Pressure Level}}
  & \multicolumn{3}{c}{\textbf{Single Level}} \\
\cmidrule(lr){4-7}\cmidrule(lr){8-10}\cmidrule(lr){11-14}\cmidrule(lr){15-17}
  &  &
  & $z500(gpm)$ & $z850(gpm)$ & $t500(K)$ & $t850(K)$ & $t2m(K)$ & $u10(m/s)$ & $v10(m/s)$
  & $z500(gpm)$ & $z850(gpm)$ & $t500(K)$ & $t850(K)$ & $t2m(K)$ & $u10(m/s)$ & $v10(m/s)$ \\
\midrule
\multirow{7}{*}{\rotatebox{90}{\textbf{Weeks 3-4}}}
  & \multicolumn{2}{c|}{\textbf{AnisoCast}} & \textbf{51.956} & \textbf{33.453} & \textbf{1.824} & \textbf{2.145} & \textbf{3.255} & \textbf{1.994} & \textbf{1.621} & \textbf{0.982} & \textbf{0.958} & \textbf{0.985} & \textbf{0.985} & \textbf{0.982} & \textbf{0.872} & \textbf{0.783} \\
  & \multicolumn{2}{c|}{w/o AE}  & 57.088 & 38.211 & 2.481 & 2.630 & 6.801 & 3.012 & 2.531 & 0.953 & 0.952 & 0.948 & 0.920 & 0.957 & 0.833 & 0.732 \\
  & \multicolumn{2}{c|}{w/o ZO} & 61.207 & 39.118 & 2.467 & 2.742 & 4.061 & 3.053 & 3.441 & 0.970 & 0.939 & 0.953 & 0.934 & 0.961 & 0.818 & 0.721 \\
  & \multicolumn{2}{c|}{w/o MO}  & 64.218 & 42.507 & 2.334 & 3.176 & 4.119 & 3.284 & 3.121 & 0.955 & 0.926 & 0.962 & 0.908 & 0.958 & 0.844 & 0.764 \\
  & \multicolumn{2}{c|}{w/ unshared ZO}  & 54.124 & 35.021 & 1.940 & 2.217 & 3.421 & 2.128 & 1.822 & 0.968 & 0.956 & 0.977 & 0.980 & 0.976 & 0.860 & 0.776 \\
  & \multicolumn{2}{c|}{w/o RN, w/ LN} & 62.948 & 40.917 & 2.412 & 2.957 & 4.047 & 3.067 & 3.302 & 0.968 & 0.932 & 0.959 & 0.924 & 0.966 & 0.856 & 0.777 \\
  & \multicolumn{2}{c|}{$\mathcal{M} \circ
    \mathcal{Z}\circ
    \mathcal{M}$} & 52.366  & 33.692  & 1.902  & 2.165  & 3.292  & 2.002  & 1.790 & 0.982  & 0.958  & 0.984  &  0.983 &  0.982 & 0.870  & 0.782  \\
\midrule
\multirow{7}{*}{\rotatebox{90}{\textbf{Weeks 5-6}}}
  & \multicolumn{2}{c|}{\textbf{AnisoCast}} & \textbf{52.268} & \textbf{33.437} & \textbf{1.830} & \textbf{2.156} & \textbf{3.231} & \textbf{1.993} & \textbf{1.606} & \textbf{0.982} & \textbf{0.958} & \textbf{0.985} & \textbf{0.985} & \textbf{0.982} & \textbf{0.873} & \textbf{0.786} \\
  & \multicolumn{2}{c|}{w/o AE}  & 56.368 & 36.271 & 2.662 & 2.920 & 5.033 & 3.541 & 1.918 & 0.961 & 0.939 & 0.955 & 0.945 & 0.960 & 0.836 & 0.726 \\
  & \multicolumn{2}{c|}{w/o ZO} & 61.543 & 40.643 & 2.721 & 2.783 & 3.889 & 3.318 & 2.952 & 0.974 & 0.936 & 0.951 & 0.953 & 0.963 & 0.821 & 0.728 \\
  & \multicolumn{2}{c|}{w/o MO}   & 62.872 & 44.019 & 2.407 & 2.792 & 3.771 & 3.544 & 2.386 & 0.961 & 0.949 & 0.960 & 0.938 & 0.971 & 0.846 & 0.759 \\
  & \multicolumn{2}{c|}{w/ unshared ZO}  & 55.213 & 36.221 & 1.960 & 2.322 & 3.501 & 2.012 & 1.708 & 0.970 & 0.950 & 0.978 & 0.980 & 0.980 & 0.862 & 0.766 \\
  & \multicolumn{2}{c|}{w/o RN, w/ LN} & 63.318 & 42.674 & 2.581 & 2.979 & 4.107 & 3.702 & 2.561 & 0.970 & 0.952 & 0.967 & 0.936 & 0.973 & 0.857 & 0.772 \\
  & \multicolumn{2}{c|}{$\mathcal{M} \circ
    \mathcal{Z}\circ
    \mathcal{M}$} & 54.447 & 34.690  & 1.843  & 2.166 & 3.341 & 2.008  & 1.705 & 0.981  & 0.957  & 0.983  & 0.983  & 0.982  & 0.872  & 0.785  \\
\bottomrule
\end{tabular}}
\end{threeparttable}
\end{table}

Table~\ref{tab:ablation} provides the full ablation results corresponding to the main-text ablation study.
We report both RMSE and ACC across the two verification windows and all evaluated variables.
The variants isolate the effects of ordered ring-wise embedding, latent zonal mixing, meridional sequence modeling, shared symmetric composition, and RMSNorm.
Overall, each component contributes to the final performance, with consistent performance degradation observed when replacing the ordered anisotropic embedding with the matched order-agnostic ring embedding. Additionally, we test an alternative symmetric ordering by replacing $\mathcal{Z}\circ\mathcal{M}\circ\mathcal{Z}$ with $\mathcal{M}\circ\mathcal{Z}\circ\mathcal{M}$ in each Aniso Block. As shown in Table~\ref{tab:ablation}, the adopted $\mathcal{Z}\circ\mathcal{M}\circ\mathcal{Z}$ ordering achieves lower RMSE for all seven variables over both weeks 3--4 and weeks 5--6, while attaining equal or higher ACC in every case. These results empirically support our choice of this ordering under the evaluated setting.

\subsection{Direct vs. Autoregressive Setting}
\label{app:direct_ar}

\begin{table}[!ht]
\centering
\caption{Comparison between autoregressive and direct S2S formulations. The direct formulation predicts the bi-weekly mean targets directly from the initial state, while the autoregressive variant recursively rolls out future states before computing the target averages.}
\label{tab:direct_ar}
\resizebox{\linewidth}{!}{
\begin{tabular}{c|c|cccc|cccc}
\toprule
\textbf{Paradigm}
& \textbf{Period}
& \textbf{z500 RMSE}
& \textbf{z850 RMSE}
& \textbf{t2m RMSE}
& \textbf{v10 RMSE}
& \textbf{z500 ACC}
& \textbf{z850 ACC}
& \textbf{t2m ACC}
& \textbf{v10 ACC} \\
\midrule
AR
& W3--4
& 56.122
& 40.112
& 5.881
& 2.421
& 0.952
& 0.951
& 0.954
& 0.730 \\
Direct
& W3--4
& \textbf{51.956}
& \textbf{33.453}
& \textbf{3.255}
& \textbf{1.621}
& \textbf{0.982}
& \textbf{0.958}
& \textbf{0.982}
& \textbf{0.783} \\
\midrule
AR
& W5--6
& 57.468
& 37.011
& 5.035
& 2.012
& 0.960
& 0.938
& 0.960
& 0.725 \\
Direct
& W5--6
& \textbf{52.268}
& \textbf{33.437}
& \textbf{3.231}
& \textbf{1.606}
& \textbf{0.982}
& \textbf{0.958}
& \textbf{0.982}
& \textbf{0.786} \\
\bottomrule
\end{tabular}
}
\end{table}

Recent direct S2S forecasting models commonly predict the target bi-weekly averages directly from the initial atmospheric state \cite{Cirt,EIMP,telepit}, rather than generating daily forecasts autoregressively and then averaging them \cite{pathak2022fourcastnet}. To verify the effect of this formulation, we compare two variants of AnisoCast: an autoregressive variant (AR), which recursively rolls out future states before computing the bi-weekly targets, and our default direct variant (Direct), which predicts the weeks 3--4 and weeks 5--6 averages in a single forward pass. The results are reported in Table~\ref{tab:direct_ar}. Direct prediction consistently outperforms the autoregressive variant across both forecast windows and all evaluated variables. This suggests that directly optimizing the evaluation-aligned bi-weekly targets is more suitable for the S2S setting, where long autoregressive rollouts can accumulate errors and degrade the final averaged forecasts.

\subsection{Performance w.r.t. the number of Aniso Blocks}
\vspace{-2mm}
\label{app:hyper_L}

We evaluate the impact of network depth by varying the number of Aniso Blocks $L \in \{4, 5, 6, 7, 8, 9\}$.
As can be seen from Figure~\ref{fig:hyper_L}, forecasting performance consistently improves as depth increases, peaking at $L=7$ where RMSE reaches a minimum for key variables.
Increasing the depth
further
yields diminishing returns or slight degradation. Thus, we select $L=7$ to balance accuracy and computational efficiency.

\begin{figure*}[!ht]
    \centering
    \includegraphics[width=\linewidth]{Figure/hyperparam_L_revised.pdf}
    \vspace{-6mm}
    \caption{RMSE and ACC variations across different numbers of Aniso Blocks ($L$).}
    \label{fig:hyper_L}
\end{figure*}

\section{Complementary Forecast Diagnostics}
\label{app:forecast_diagnostics}

\subsection{Spectral Energy Distribution Analysis}
\label{app:SpecRes}
\vspace{-2mm}

\begin{figure*}[!ht]
    \centering
    \includegraphics[width=\textwidth]{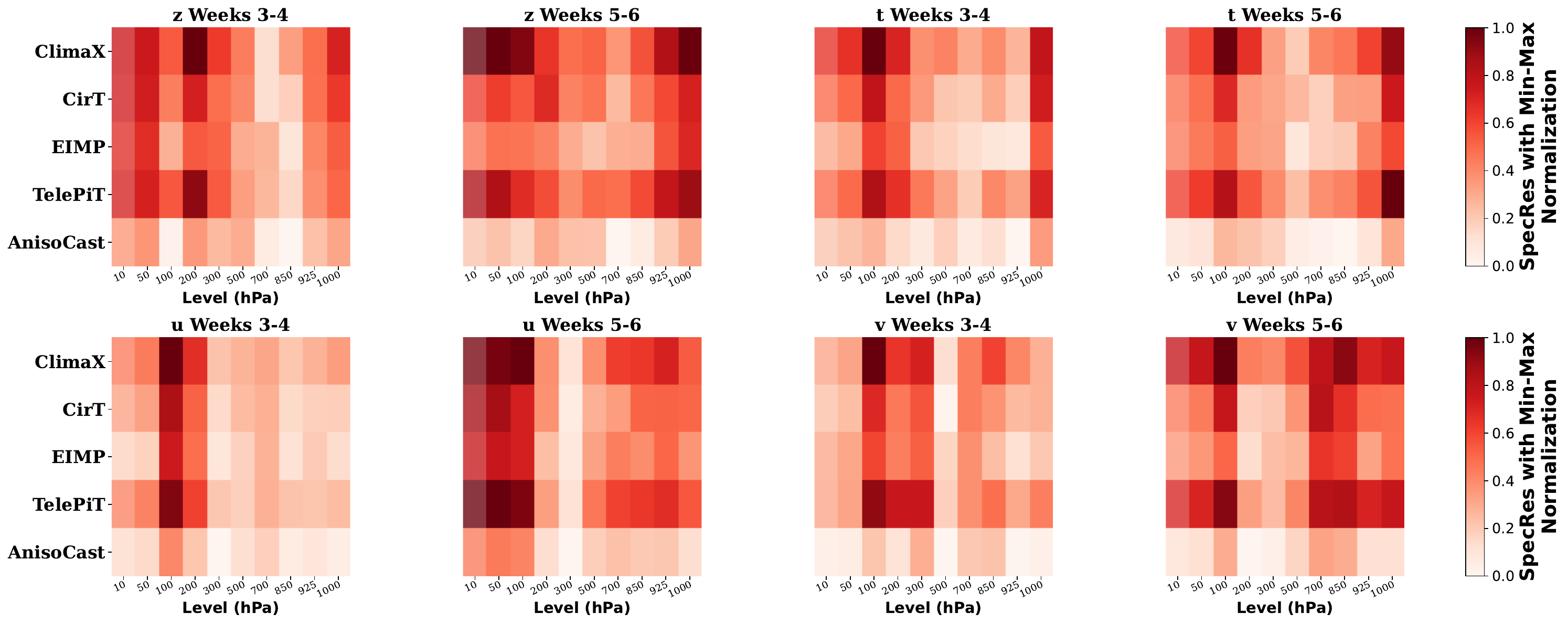}
    \caption{SpecRes comparison across pressure levels.}
    \label{fig:specres_levels}
\end{figure*}

We further assess spectral energy distribution accuracy using SpecRes~\citep{telepit} across all pressure levels for the prognostic variables $z$, $t$, $u$, and $v$ under weeks~3--4 and weeks~5--6 forecasts. SpecRes quantifies the root of the expected squared residual in the spectral domain, i.e., $\text{SpecRes}=\sqrt{\mathbb{E}_{k}\!\left[(\hat{S}'(k)-S'(k))^{2}\right]}$, where $\hat{S}'(k)$ and $S'(k)$ denote the normalized power spectra of prediction and target, respectively.
Therefore, lower SpecRes values indicate a more accurate preservation of spectral energy distribution.
As shown in Figure~\ref{fig:specres_levels}, AnisoCast consistently exhibits substantially lower SpecRes values than all competing models over nearly all vertical levels, indicating a more faithful preservation of spectral energy distribution at long forecast lead times. This indicates closer agreement with the target power spectra,
complementing the grid-space RMSE evaluation.

\subsection{Vertical Profile Analysis}
\label{app:latidunal}
We further analyze the model's performance across the vertical atmospheric column to ensure physical consistency at different altitudes. Figure~\ref{fig:vertical_rmse_uv} compares the RMSE of  $u, v$ against data-driven baselines, while Figure~\ref{fig:vertical_rmse_zt} benchmarks $z, t$ against operational NWP systems. AnisoCast achieves the lowest normalized RMSE (indicated by the lightest colors) across almost all 10 pressure levels, across ten pressure levels from 10 to 1000~hPa. This vertical consistency is particularly notable in comparison to models like OmniCast~\citep{nguyen2025omnicast} and CirT~\citep{Cirt}, which exhibit degradation in the upper atmosphere. The results demonstrate that AnisoCast's anisotropic embedding and operator design effectively capture the distinct dynamics of different atmospheric layers without requiring level-specific fine-tuning.

\begin{figure*}[!ht]
    \centering
    \includegraphics[width=\linewidth]{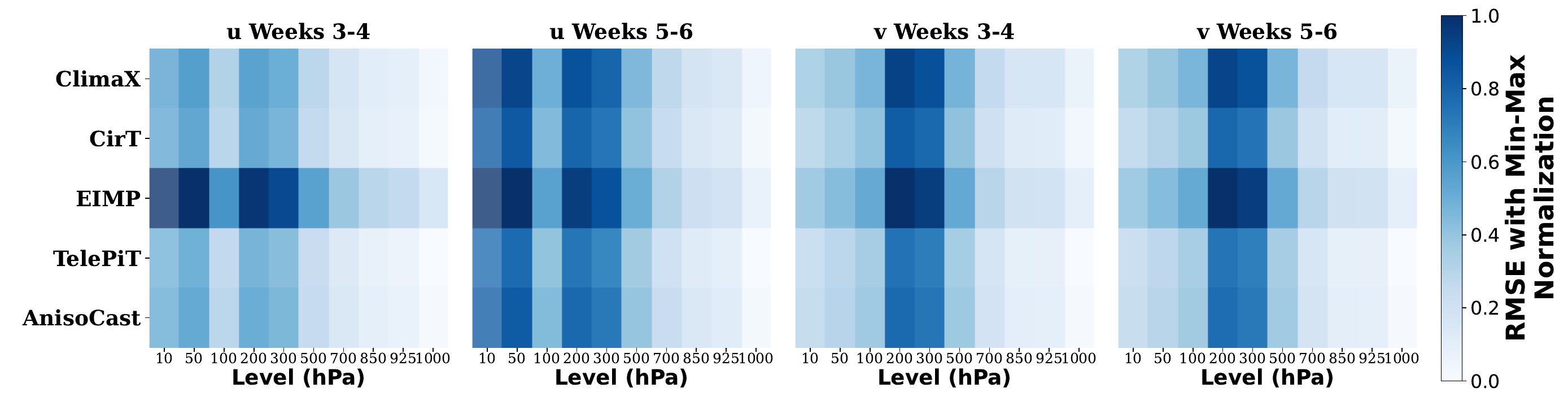}
    \vspace{-6mm}
    \caption{RMSE comparison on variables $u, v$ across different pressure levels.}
    \label{fig:vertical_rmse_uv}
\end{figure*}

\begin{figure*}[!ht]
    \centering
    \includegraphics[width=\linewidth]{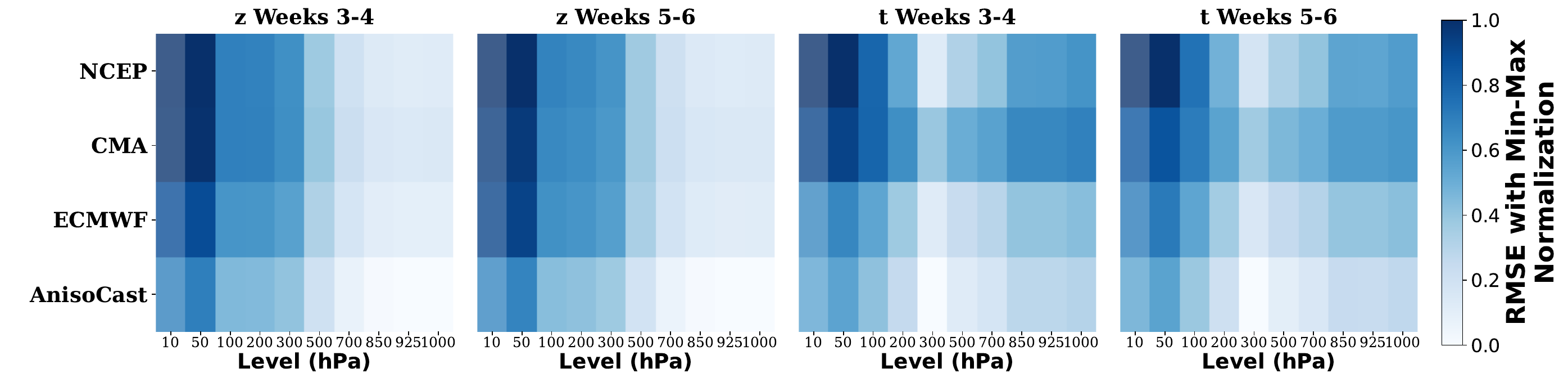}
    \vspace{-6mm}
    \caption{RMSE comparison on variables $z$ and $t$ across different pressure levels.}
    \label{fig:vertical_rmse_zt}
\end{figure*}

\subsection{Daily Visualization and Temporal Stability}
\label{app:dailyvisal}
To probe fine-grained temporal stability, we trained AnisoCast and other baselines to output the full daily trajectory. Figure~\ref{fig:dailyw34} and Figure~\ref{fig:dailyw56} visualize the day-by-day evolution of RMSE and ACC for weeks 3--4 (days 15--28) and weeks 5--6 (days 29--42), respectively. In the weeks 3--4 window (Figure~\ref{fig:dailyw34}), AnisoCast consistently maintains superior performance across all variables. While baselines suffer from rapid degradation over time (manifested as steep error slopes), AnisoCast exhibits a remarkably flat error growth rate. This advantage becomes even more pronounced in the extended weeks 5--6 window (Figure~\ref{fig:dailyw56}). For highly chaotic surface variables like $v10$, the performance gap between AnisoCast and the nearest competitor widens significantly as the lead time increases. These results support the temporal performance of the
daily-output variant.

\begin{figure*}[!ht]
    \centering
    \includegraphics[width=\linewidth]{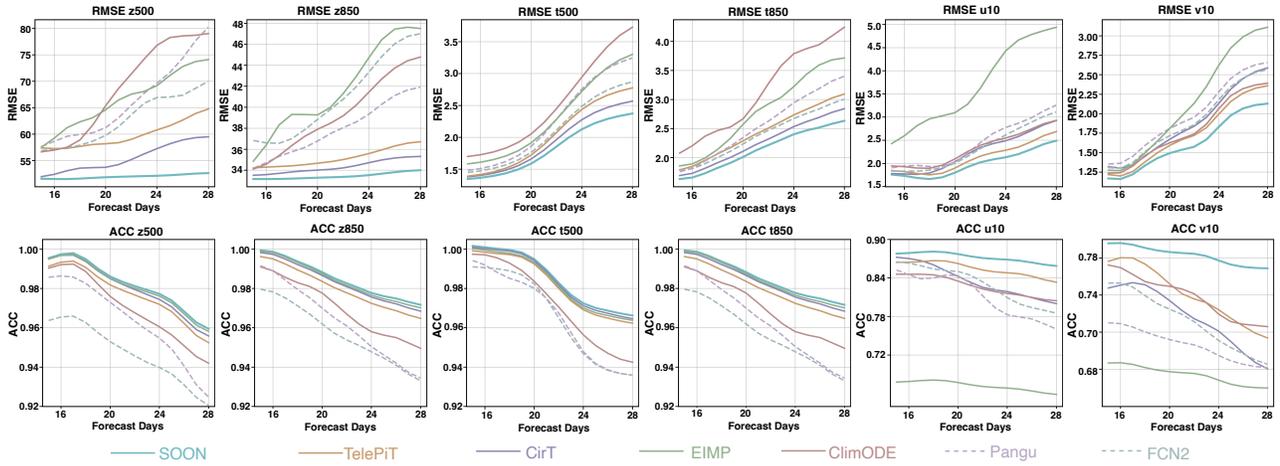}
    \vspace{-6mm}
    \caption{Daily RMSE and ACC variations across critical variables during weeks 3--4 (Days 15--28).}
    \label{fig:dailyw34}
\end{figure*}

\begin{figure*}[!ht]
    \centering
    \includegraphics[width=\linewidth]{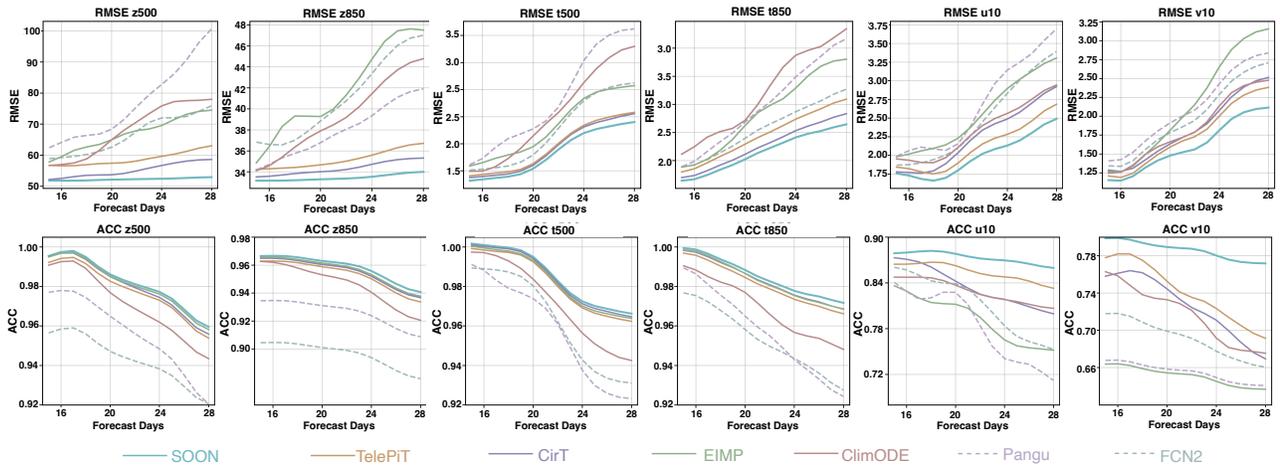}
    \vspace{-6mm}
    \caption{Daily RMSE and ACC variations across critical variables during weeks 5--6 (Days 29--42).}
    \label{fig:dailyw56}
\end{figure*}

\subsection{Additional Global Visualizations}

\label{app:globalvis}
We provide expanded global error maps to visually assess spatial performance distributions across different variables and lead times. Figure~\ref{fig:t850_w5-6} through Figure~\ref{fig:z850_w5-6} display the RMSE distributions for $t850, t500, z500,$ and $z850$ during weeks 3--4 and weeks 5--6. Consistent with the findings in the main text, AnisoCast exhibits deep blue patterns globally, indicating low error. The maps show pronounced error differences in several high-latitude regions, with AnisoCast exhibiting lower errors
in the highlighted areas across both forecast windows. These comparisons assess spatial forecast accuracy rather than directly identifying the physical cause of the improvements.

\begin{figure*}[!t]
    \centering
    \includegraphics[width=0.86\linewidth, trim=6 12 6 6, clip]{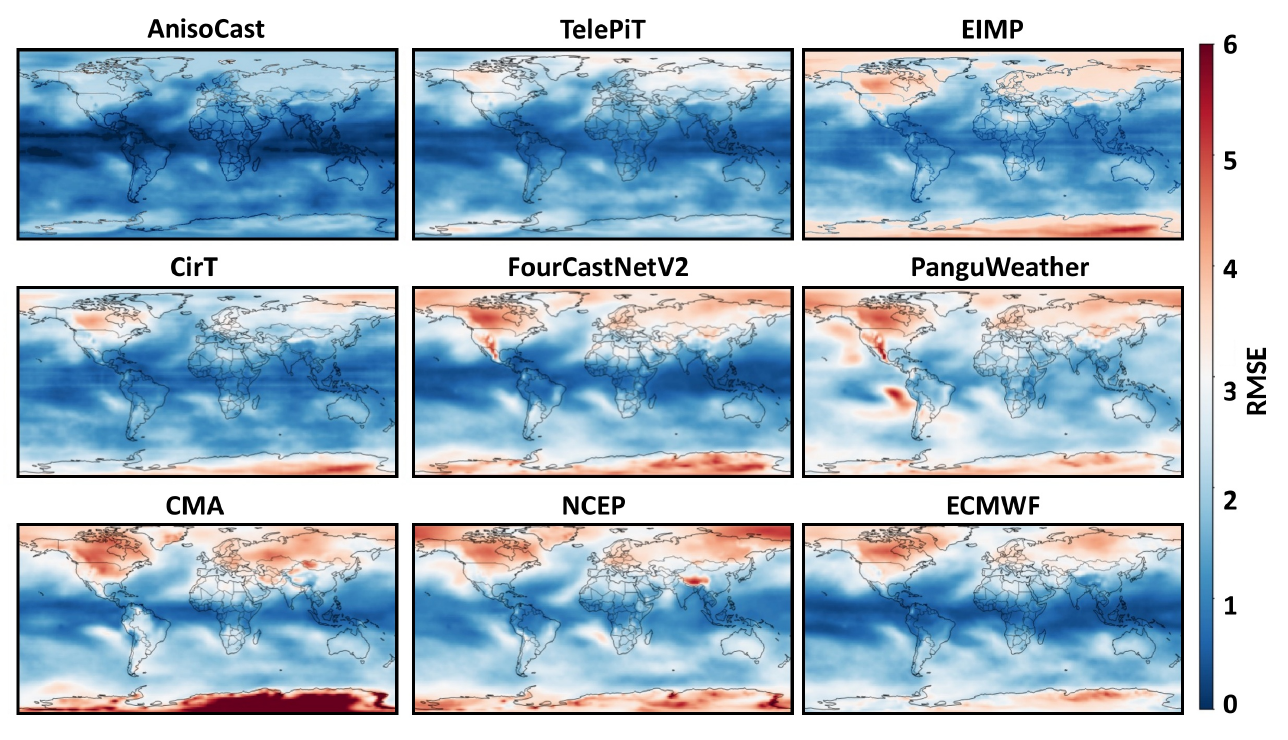}
    \vspace{-2mm}
    \caption{Global RMSE distribution for $t850$ at weeks 5--6 lead time.}
    \label{fig:t850_w5-6}
\end{figure*}

\begin{figure*}[!t]
    \centering
    \includegraphics[width=0.88\linewidth, trim=6 12 6 6, clip]{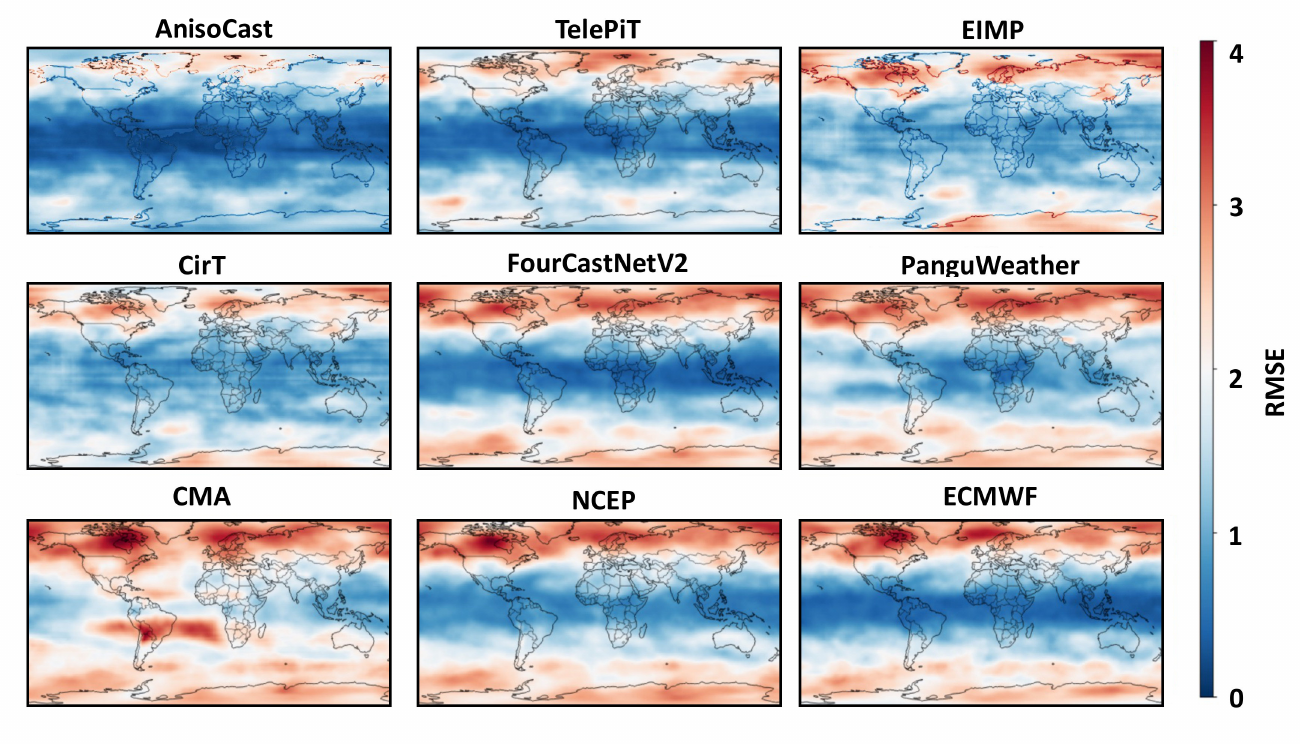}
    \vspace{-2mm}
    \caption{Global RMSE distribution for $t500$ at weeks 3--4 lead time.}
    \label{fig:t500_w3-4}
\end{figure*}

\begin{figure*}[!t]
    \centering
    \includegraphics[width=0.88\linewidth, trim=6 12 6 6, clip]{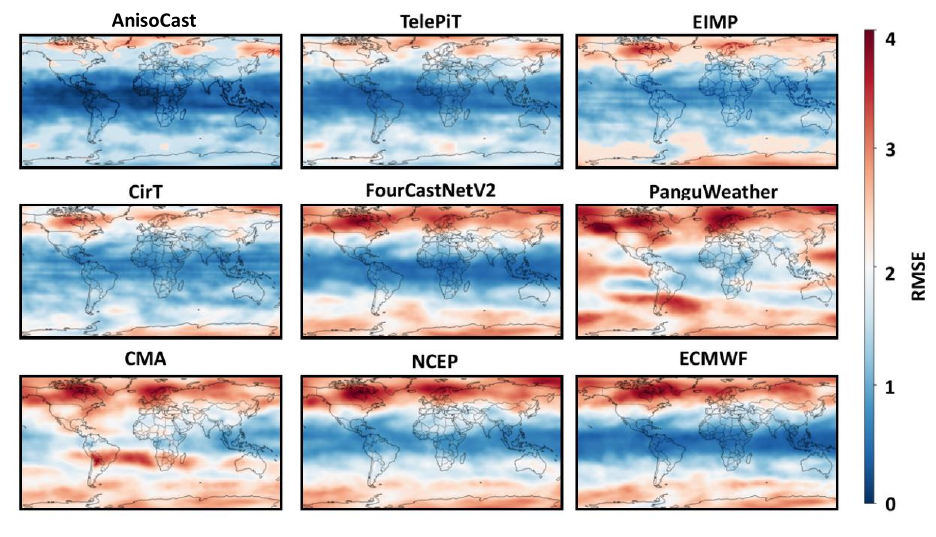}
    \vspace{-2mm}
    \caption{Global RMSE distribution for $t500$ at weeks 5--6 lead time.}
    \label{fig:t500_w5-6}
\end{figure*}

\begin{figure*}[!t]
    \centering
    \includegraphics[width=0.88\linewidth, trim=6 12 6 6, clip]{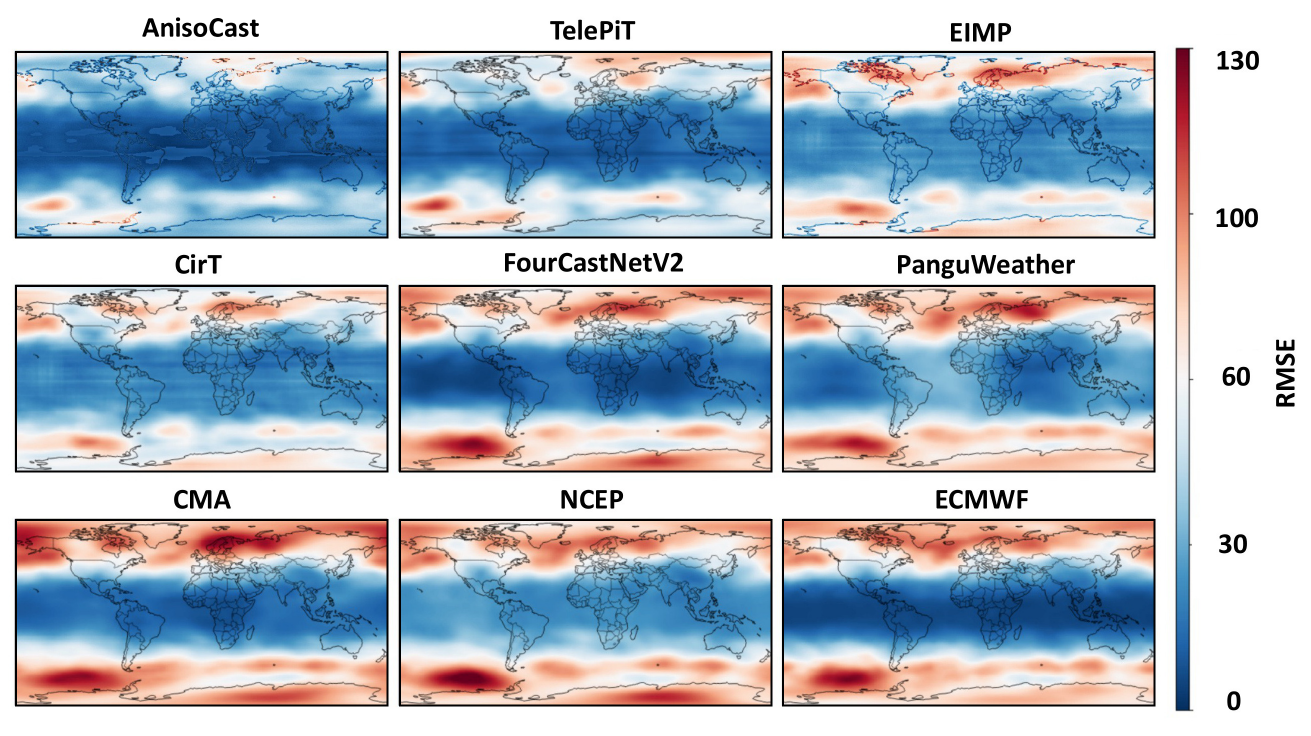}
    \vspace{-2mm}
    \caption{Global RMSE distribution for $z500$ at weeks 3--4 lead time.}
    \label{fig:z500_w3-4}
\end{figure*}

\begin{figure*}[!t]
    \centering
    \includegraphics[width=0.88\linewidth, trim=6 12 6 6, clip]{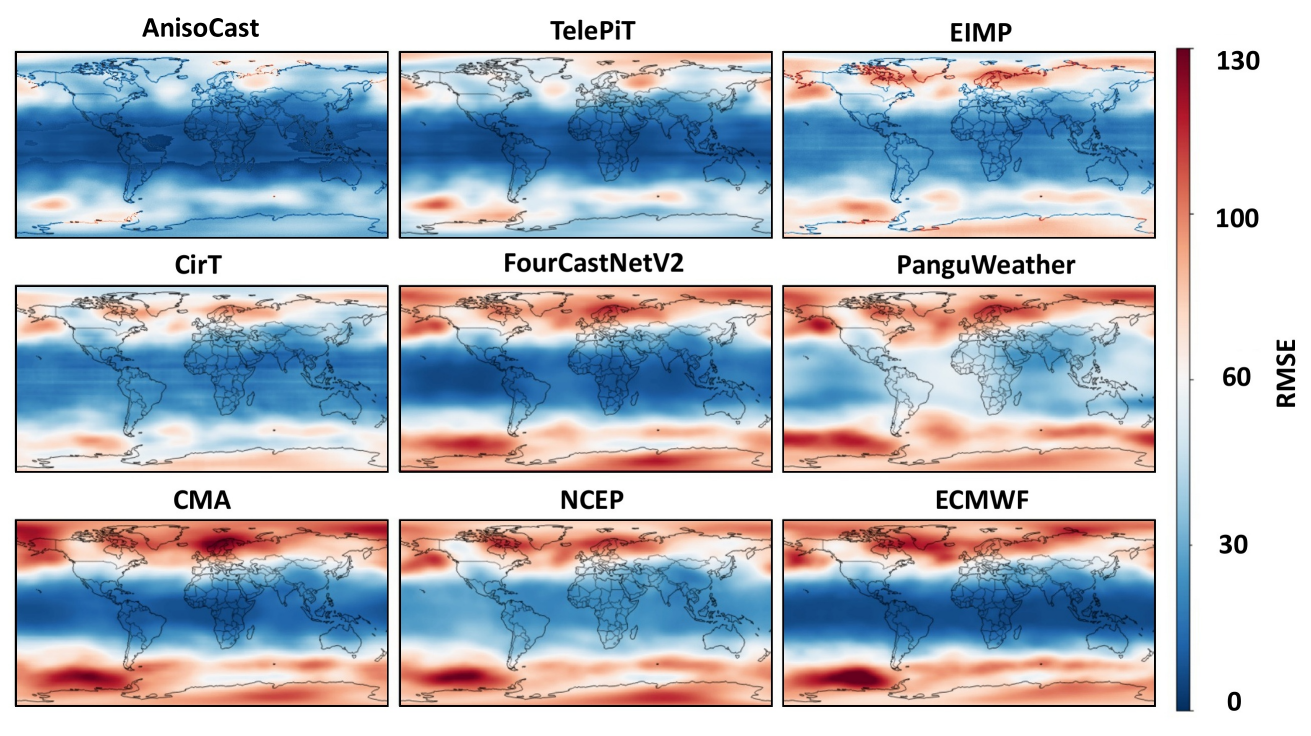}
    \vspace{-2mm}
    \caption{Global RMSE distribution for $z500$ at weeks 5--6 lead time.}
    \label{fig:z500_w5-6}
\end{figure*}

\begin{figure*}[!t]
    \centering
    \includegraphics[width=0.88\linewidth, trim=6 12 6 6, clip]{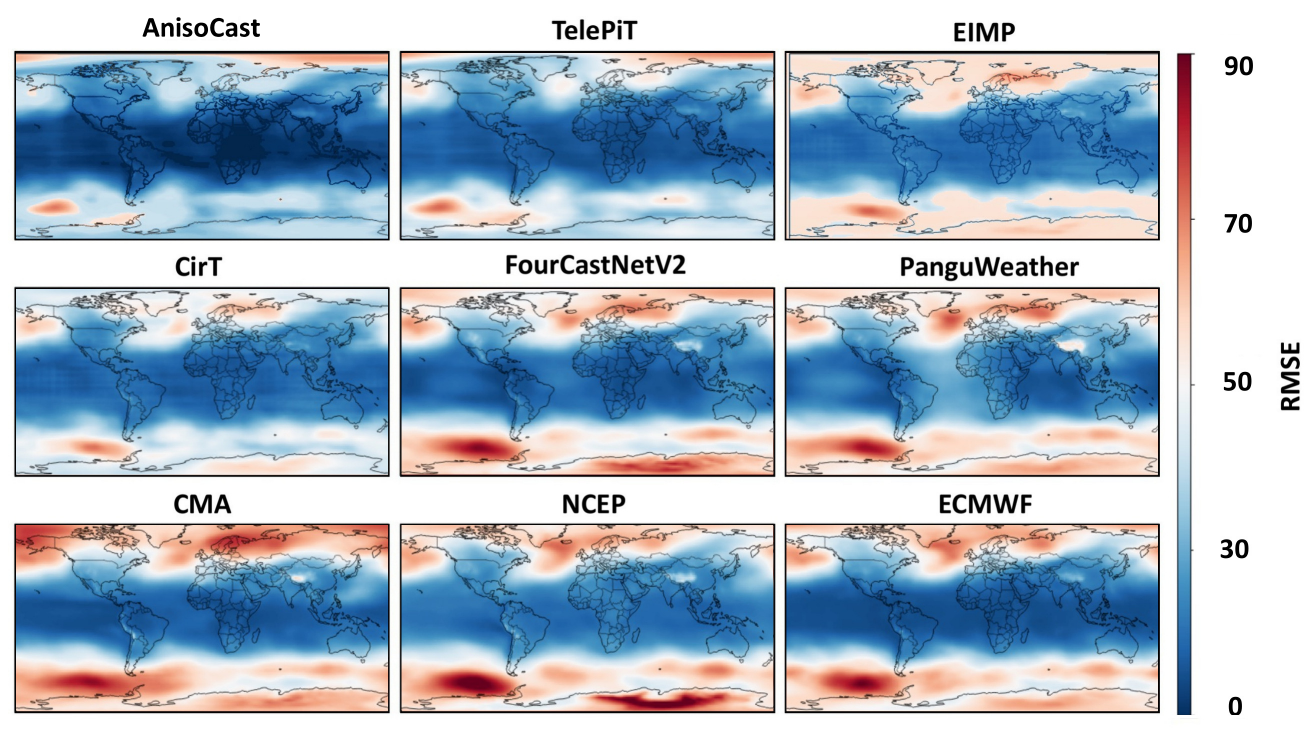}
    \vspace{-2mm}
    \caption{Global RMSE distribution for $z850$ at weeks 3--4 lead time.}
    \label{fig:z850_w3-4}
\end{figure*}

\begin{figure}[!t]
    \centering
    \includegraphics[
        width=0.88\linewidth,
        trim=6 12 6 6,
        clip
    ]{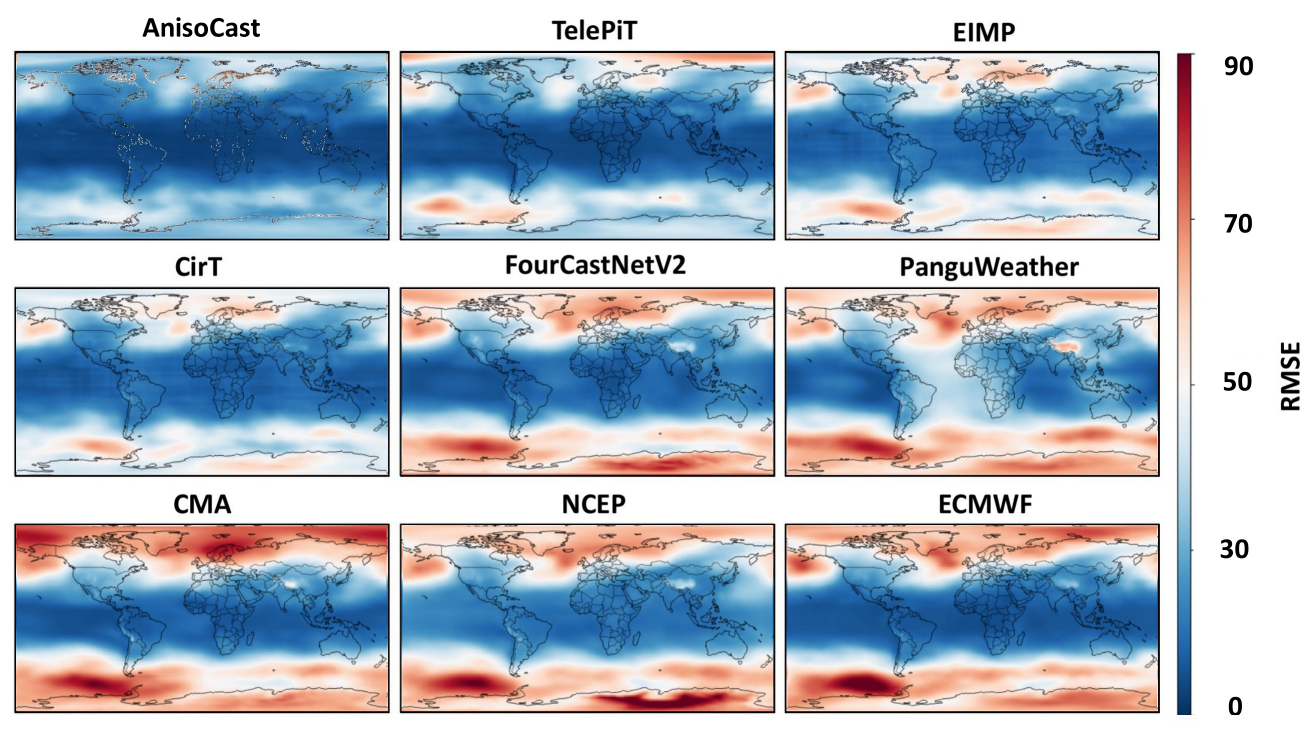}
    \vspace{-2mm}
    \caption{Global RMSE distribution for $z850$ at weeks 5--6 lead time.}
    \label{fig:z850_w5-6}
\end{figure}

\section{Limitations}
\label{appendix:limitations}

While AnisoCast achieves strong performance under the ERA5-based direct
bi-weekly S2S benchmark, it currently relies on a deterministic regression
objective and therefore does not explicitly quantify forecast uncertainty
or provide probabilistic skill scores. Our evaluation focuses on grid-field
accuracy metrics and selected deterministic diagnostics; it does not yet
include a full set of process-oriented S2S diagnostics, such as MJO
propagation skill and teleconnection fidelity. Future work will extend AnisoCast
toward probabilistic or ensemble forecasting, process-oriented S2S
verification, and higher-resolution operational evaluation.
\end{document}